\documentclass[10pt,reqno]{article}
\usepackage{fullpage} 
\usepackage{amssymb} 
\usepackage{graphicx} 
\usepackage{amsmath} 
\usepackage{amsthm,amssymb,latexsym} 
\usepackage{amstext, amsfonts,amsbsy} 
\usepackage{enumerate} 
\usepackage{upgreek} 
\usepackage{color} 
\usepackage{accents} 
\usepackage{mathtools} 
\usepackage{tikz}
\usetikzlibrary{matrix,graphs,arrows,positioning,calc,decorations.markings,shapes.symbols}
\usepackage[pdftex,bookmarks,colorlinks,breaklinks]{hyperref}  
\definecolor{dullmagenta}{rgb}{0.4,0,0.4}   
\definecolor{darkblue}{rgb}{0,0,0.4}
\hypersetup{linkcolor=red,citecolor=blue,filecolor=dullmagenta,urlcolor=darkblue} 
\usepackage{eucal}

\newtheorem{theorem}{Theorem}

\newtheorem{lemma}[theorem]{Lemma}

\newcommand{\Pain}[1]{\text{P}_{\mathrm{#1}}}
\newcommand{\dPain}[1]{\text{P}\left(\mathrm{#1}\right)}

\theoremstyle{definition}

\theoremstyle{remark}
\newtheorem{remark}[theorem]{Remark}

\numberwithin{equation}{section}

\begin{document}

{\noindent\Large\bf Recurrence coefficients for discrete orthogonal polynomials with hypergeometric 
weight and discrete Painlev\'e equations 
%
%
}
\medskip
\begin{flushleft}

\textbf{Anton Dzhamay}\\
School of Mathematical Sciences, The University of Northern Colorado, Greeley, CO 80526, USA\\
E-mail: \href{mailto:adzham@unco.edu}{\texttt{adzham@unco.edu}}\\[5pt]

\textbf{Galina Filipuk}\\
Faculty of Mathematics, Informatics and Mechanics, University of Warsaw, Banacha 2, Warsaw, 02-097, Poland\\
E-mail: \href{mailto:filipuk@mimuw.edu.pl}{\texttt{filipuk@mimuw.edu.pl}}\\[5pt] 

\textbf{Alexander Stokes}\\
Department of Mathematics, University College London, Gower Street, London, WC1E 6BT, UK
E-mail: \href{mailto:alexander.stokes.14@ucl.ac.uk}{\texttt{alexander.stokes.14@ucl.ac.uk}}\\[8pt]



\end{flushleft}

%
%
%
%

%
\date{}

\begin{abstract}

Over the last decade it has become clear that discrete Painlev\'e
equations appear in a wide range of important mathematical and physical problems. 
Thus, the question of recognizing a given non-autonomous recurrence as a discrete 
Painlev\'e equation and determining its type according to Sakai's classification 
scheme, understanding whether it is equivalent to some known (model) example, and 
especially finding an explicit change of coordinates transforming it to such an example, 
becomes one of the central ones. Fortunately, Sakai's geometric theory provides an 
almost algorithmic procedure for answering this question. In this paper we illustrate 
this procedure by studying an example coming from the theory of discrete orthogonal 
polynomials. There are many connections between orthogonal polynomials and Painlev\'e
equations, both differential and discrete. In particular, often the coefficients of 
three-term recurrence relations for discrete orthogonal polynomials can be expressed 
in terms of solutions of  discrete Painlev\'e equations. In this work we study discrete 
orthogonal polynomials with general hypergeometric weight and show that their recurrence 
coefficients satisfy, after some change of variables, the standard discrete Painlev\'e-V 
equation. We also provide an explicit change of variables transforming this equation to 
the standard form. 

%
%
%

\end{abstract}

%

\section{Introduction} 
\label{sec:introduction}

In describing interesting physical and mathematical models we often rely on various special functions, such as Airy or Bessel functions.
Such functions satisfy certain \emph{linear} ordinary differential equations, and over a hundred years ago P.~Painlev\'e became interested
in the question of whether it may be possible to define purely \emph{nonlinear special functions} as solutions of \emph{nonlinear} ordinary
differential equations. As usual, the nonlinear case is quite subtle, since solutions of nonlinear differential equations do not satisfy 
the superposition principle, and in general, it may not even be possible to define the notion of a \emph{general solution} since solutions 
can develop unexpected singularities that depend not just on the equation, but also on the
initial conditions. Nevertheless, this line of reasoning led Painlev\'e to define a property of a nonlinear
ordinary differential equation (essentially the absence of \emph{movable}, i.e., dependent on initial conditions, essential 
critical points) that guarantees the existence of a general solution; this property is now called the \emph{Painlev\'e Property}.
Painlev\'e and his student B.~Gambier then studied a large class of algebraic second-order differential equations that satisfy this property and 
found that, in addition to equations that are linear or can be reduced to linear, there are six new families of equations that 
are now called \emph{Painlev\'e equations} $\Pain{I},\dots, \Pain{VI}$. Solutions of these equations, the so-called \emph{Painlev\'e transcendents}, 
are indeed new \emph{purely nonlinear special functions}. Over the last fifty years Painlev\'e transcendents have been playing an 
increasingly important role in the  description of many nonlinear models, from Quantum Cohomology to the theory of Random Matrices. 
Probably the most important example is the famous Tracy-Widom distribution from Random Matrix Theory that can be expressed in terms of 
the Hastings-McLeod solution of the Painlev\'e-II equation. 

The theory of discrete Painlev\'e equations is much more recent. These equations were originally defined as second-order 
discrete non-linear equations (or second-order nonlinear recurrence relations) that become one of the differential Painlev\'e 
equations in a continuous limit. The intensive study of these equations began in the early 1990's \cite{RamGraHie:1991:DVOTPE} 
and many examples were obtained in the works of B.~Grammaticos, A.~Ramani, and their collaborators 
by the application of the \emph{singularity confinement} criterion to deautonomizations of known discrete dynamical systems,
see the review \cite{GraRam:2004:DPER} and references therein. Discrete Painlev\'e equations were also studied from the perspective 
of the \emph{representation theory of affine Weyl groups} in a series of papers by  M.~Noumi, Y.~Yamada, and their collaborators, 
see, e.g, \cite{NouYam:1998:AWGDDSAPE}. In 2001 H.~Sakai, in his seminal paper \cite{Sak:2001:RSAWARSGPE} that used techniques
from \emph{birational algebraic geometry}, gave the definitive classification scheme for discrete Painlev\'e equations and clarified 
the relationship between discrete and differential Painlev\'e equations. Since then the theory of discrete Painlev\'e equations has
reached a certain level of maturity. We know many examples of discrete Painlev\'e equations, their properties, 
special solutions for certain parameter values, Lax pairs, various degenerations, etc.; the recent survey paper \cite{KajNouYam:2017:GAOPE} 
is both an excellent introduction and a comprehensive overview of the present theory of discrete Painlev\'e equations.

Moreover, there is an increasing body of evidence that discrete Painlev\'e equations, similar to their differential counterpart, appear in 
a wide variety of important applied problems, such as the computations of gap probabilities \cite{Bor:2003:DGPADPE} of various ensembles 
in the emerging field of \emph{integrable probability} \cite{BorGor:2016:LIP}, or in describing recurrence coefficients of semi-classical 
orthogonal polynomials \cite{Van:2018:OPAPE}, and many others. To make a connection between an applied problem and the wealth of known
results, it is then important to be able to answer the following sequence of questions:
\begin{enumerate}[(a)]
	\item Suppose one obtains a certain non-linear second order recurrence relation. Does this recurrence fit into the discrete Painlev\'e framework, 
	i.e., into Sakai's classification scheme?
	\item If so, what is the type of this equation, i.e., the type of its algebraic surface in Sakai's classification?
	\item After the type of the equation is determined, the next question is whether it is equivalent to any known examples of equations of the same type.
	In general, there are infinitely many non-equivalent discrete Painlev\'e equations, but usually some simplest forms
	of such equations are well-known. For example, canonical examples of equations of each type are listed in \cite{KajNouYam:2017:GAOPE},
	see also Sakai's original paper \cite{Sak:2001:RSAWARSGPE}.
	\item Finally, if the equation is indeed equivalent to a canonical form of some discrete Painlev\'e equation, how to find an explicit
	change of variables transforming one equation into the other. In particular, answering this question requires matching of various
	parameters in the applied problem with parameters in the standard form of this discrete Painlev\'e equation. Note also that being able to do this may also 
	result in uncovering new connections between very different problems.
\end{enumerate}

Fortunately, the algebro-geometric theory of Painlev\'e equations provides us with a powerful set of tools and essentially a near-algorithmic procedure
to answer exactly these questions. Unfortunately, the necessary mathematical background to master this theory, such as birational algebraic geometry, 
the representation theory of affine Weyl groups, and the word equivalence problem in groups, 
is often quite different from that of the researchers working with applied problems, and so the learning curve can feel steep. Nevertheless, we believe that 
it is still possible to learn, at least on a computational level first, the essentials of how to approach these questions. Thus, the purpose of the present paper
is to illustrate the above procedure in detail using one concrete example, hoping that anyone interested would then be able to make necessary changes
to adjust this procedure for a different example.

The problem that we consider belongs to the theory of \emph{orthogonal polynomials}. 
In fact, the relationship between discrete Painlev\'e equations and orthogonal polynomials is much older than the actual \emph{definition} of a discrete Painlev\'e
equation --- the first example of a discrete Painlev\'e-I equation originally appeared in the work of Shohat \cite{Sho:1939:DEOP}. There are many 
connections between recurrence coefficients of semi-classical orthogonal polynomials and solutions of  Painlev\'e equations, both discrete and 
differential (see, for instance, \cite{Van:2018:OPAPE} and numerous references therein).

Let $\{p_{n}(x) = \gamma_{n} x^{n} + \cdots \}$ be the collection of polynomials that are orthonormal on the set $\mathbb{N}=\{0,1,2,\ldots\}$
of non-negative integers with respect to the \emph{hypergeometric weight} $w_{k}$,
\begin{equation}\label{eq:hyp-weight}
	\sum_{k=0}^\infty p_n(k)p_m(k) w_k = \delta_{m,n}, \qquad 
	w_k = \frac{(\alpha)_k (\beta)_k}{(\gamma)_k k!} c^k, \quad \alpha,\beta,\gamma >0,\ 0 < c < 1,
\end{equation}
where $(\cdot)_k$ is the usual Pochhammer symbol and $\delta_{m,n}$ is the Kronecker delta. This collection of polynomials is known  
as the discrete orthogonal polynomials with hypergeometric weights, since the moments of this weight function are given in terms of the Gauss hypergeometric function ${}_2F_1(\alpha,\beta;\gamma;c)$ and its derivatives; it has been studied in \cite{FilVan:2018:DOPWHWAPV}. These polynomials satisfy the three term recurrence relation
\begin{equation}
	xp_n(x) = a_{n+1}p_{n+1}(x) + b_n p_n(x) + a_n p_{n-1}(x),
\end{equation}
where  $a_0=0$. The coefficients $a_n$ and $b_n$ are called the \emph{recurrence coefficients} \cite{Chi:1978:AITOP, Ism:2005:CAQOPIOV, Sze:1967:OP}. 
Note that the corresponding \emph{monic} orthogonal polynomials $P_n = p_n/\gamma_n$ satisfy a similar three term recurrence relation
\begin{equation}
	xP_n(x) = P_{n+1}(x) + b_n P_n(x) + a_n^2 P_{n-1}(x). 
\end{equation}
In \cite{FilVan:2018:DOPWHWAPV} it was shown that these recurrence coefficients $\{a_{n},b_{n}\}$, as functions of the discrete variable $n$, satisfy, after some 
change of variables, a system of non-linear difference equations and as functions
of the continuous parameter $c$ of the hypergeometric weight, they satisfy the differential Toda system. 
From the differential and discrete systems one can obtain a differential equation, which in turn can be reduced to the 
$\sigma$-form of the sixth Painlev\'e equation.
In \cite{HuFilChe:2019:DDERCOPHWBTPE}, using a direct computation, it was then shown that the discrete system is a composition of B\"acklund transformations of the sixth Painlev\'e equation. In the present paper we give a geometric explanation of this result, show that the discrete system is in fact equivalent 
to the standard discrete Painlev\'e-V equation, and provide an explicit change of variables achieving that.

To be more specific, let us introduce two new variables $x_{n}$ and $y_{n}$ parameterizing the recurrence coefficients $a_{n}^{2}$ and $b_{n}$ via
\begin{align}\label{eq:xy-vars}
   a_n^2\frac{1-c}{c}&=y_n+\sum_{k=0}^{n-1}x_{k}+\frac{n(n+\alpha+\beta-\gamma-1)}{1-c},&\quad 
   b_{n}&=x_{n}+\frac{n+(n+\alpha+\beta)c-\gamma}{1-c}.
\end{align} 
It was shown in \cite[Theorem 3.1]{FilVan:2018:DOPWHWAPV} that $x_{n}$, $y_{n}$, $n\in \mathbb{N}$, satisfy the  first-order system of non-linear 
non-autonomous difference equations
\begin{align}
\begin{split}\label{eq:yn-evol}
(y_{n}-&\alpha\beta+(\alpha+\beta+n)x_{n}-x_{n}^{2})(y_{n+1}-\alpha\beta+(\alpha+\beta+n+1)x_{n}-x_{n}^{2})\\
       &=\frac{1}{c}(x_{n}-1)(x_{n}-\alpha)(x_{n}-\beta)(x_{n}-\gamma),
\end{split}\\
\begin{split}\label{eq:xn-evol}
(x_{n}+&\mathfrak{Y}_{n})(x_{n-1}+\mathfrak{Y}_{n})\\
&=\frac{(y_{n}+n\alpha)(y_{n}+n\beta)(y_{n}+n\gamma-(\gamma-\alpha)(\gamma-\beta))(y_{n}+n-(1-\alpha)(1-\beta))}{(y_{n}(2n+\alpha+\beta-\gamma-1)+n((n+\alpha+\beta)(n+\alpha+\beta-\gamma-1)-\alpha\beta+\gamma))^{2}},
\end{split}
\intertext{where $\alpha, \beta, \gamma, c$ are the parameters of the hypergeometric weight $w_{k}$ in \eqref{eq:hyp-weight} and }
\mathfrak{Y}_{n}&=\frac{y_{n}^{2}+y_{n}(n(n+\alpha+\beta-\gamma-1)-\alpha\beta+\gamma)-\alpha\beta n(n+\alpha+\beta-\gamma-1)}{y_{n}(2n+\alpha+\beta-\gamma-1)+n((n+\alpha+\beta)(n+\alpha+\beta-\gamma-1)-\alpha\beta+\gamma)}.\label{eq:Yn}
\intertext{The initial conditions for this recurrence are given by}
\begin{split}\label{eq:xy-ic}
	x_0=\frac{\alpha\beta c}{\gamma}\frac{{}_2F_{1}(\alpha+1,\beta+1;\gamma+1;c)}{{}_2F_{1}(\alpha,\beta;\gamma;c)}+\frac{(\alpha+\beta)c-\gamma}{c-1},\quad y_0=0.
\end{split}
\end{align}

For the hypergeometric weights the connection with the
$\sigma$-form of the sixth Painlev\'e equation (with independent variable $c$) is known (see \cite[Theorem 5.1]{FilVan:2018:DOPWHWAPV}). The essential role is played by the Toda system for the recurrence coefficients (see, e.g., \cite[\S 2.8]{Ism:2005:CAQOPIOV} or \cite[\S 3.2.2]{Van:2018:OPAPE}). For the hypergeometric weights, it is  given by
\begin{eqnarray}\label{(3.1)}
     c \frac{d}{dc} a_n^2 &=& a_n^2(b_n-b_{n-1}), \qquad n \geq 1,  \\ \label{(3.2)}
     c \frac{d}{dc} b_n &=&  a_{n+1}^2 - a_n^2, \qquad n \geq 0.
\end{eqnarray}
It is proved in \cite[Theorem 5.1]{FilVan:2018:DOPWHWAPV} that a simple linear change of variable transforms $S_n=\sum_{k=0}^{n-1}x_k$ into the solutions of the $\sigma$-form of the sixth Painlev\'e equation. Knowing $S_n$ one can find $x_n,$ $y_n$ and, hence,  the recurrence coefficients $a_n^2$, $b_n$ in terms of $S_n$ and its derivatives. Moreover, it is shown in \cite{HuFilChe:2019:DDERCOPHWBTPE} that the differential equation for $x_n$ can be directly reduced to the sixth Painlev\'e equation.

Our main objective for this paper is to illustrate the general process of identifying a discrete dynamical system as a discrete Painlev\'e equation and explicitly 
rewriting it in some standard form, using equations (\ref{eq:yn-evol}--\ref{eq:xn-evol}) as an example of an applied system. This process consists of the following steps, where
we assume that we indeed have some discrete Painev\'e equation, otherwise the process will terminate at some step.

\begin{enumerate}[(Step 1)]
	\item \textbf{Identify the singularity structure of the problem.} For that, if necessary, rewrite our recurrence equation as a system of
			two first-order recurrences, $(x_{n+1},y_{n+1}) = \psi^{(n)}(x_{n},y_{n})$. 
			The mapping $\psi^{(n)}: \mathbb{C}^{2}\to \mathbb{C}^{2}$ should be 
			a birational mapping that may depend on various parameters, including the iteration step $n$ that we consider to be generic.
			Then compactify the configuration space from $\mathbb{C}^{2} = \mathbb{C}\times \mathbb{C}$ to 
			$\mathbb{P}^{1} \times \mathbb{P}^{1}$. Find the base points of the mapping and resolve
			them using the blowup procedure. Continue doing that until all base points for both $\psi^{(n)}$ and $(\psi^{(n)})^{-1}$
			are resolved (for discrete Painlev\'e equations this process should 
			terminate in finitely many steps). Thus, we get an isomorphism of resulting rational algebraic surfaces, 
			$\psi^{(n)}: \mathcal{X}_{n}\xrightarrow{\simeq}\mathcal{X}_{n+1}$. In making this computation, it is important to keep in mind that 
			positions of base points usually evolve with the mapping, so one needs to be careful distinguishing between the points in the domain 
			and the points in the range. We also remark that sometimes
			the singularity structure can be seen before even studying the dynamics; e.g., singularities can occur as a result of a
			parameterization of some moduli space appearing in the problem, as in \cite{DzhKni:2019:QPEDQE}, for example.
	\item \textbf{Linearize the mapping on $\operatorname{Pic}(\mathcal{X})$.} This can be done explicitly in relatively simple cases, such as the 
			present example. But sometimes the evolution mapping can be too complicated even for a Computer Algebra System. In this case, it may be 
			possible to deduce the action of the mapping on $\operatorname{Pic}(\mathcal{X})$ from the knowledge of parameter evolution via the 
			\emph{Period Map}, see \cite{DzhKni:2019:QPEDQE} for an example of such a computation.	
	\item \textbf{Determine the surface type, according to Sakai's classification scheme.} For a discrete Painlev\'e equation, although 
			the positions of base points may evolve, the \emph{configuration} will stay fixed, and so the surfaces $\{\mathcal{X}_{n}\}$ will 
			all have the same type. There should be \emph{eight} such base points; those points will lie on some (generically unique) 
			biquadratic curve on $\mathbb{P}^{1}\times \mathbb{P}^{1}$ (i.e., a curve whose defining polynomial, when written in a coordinate chart, has
			bi-degree $(2,2)$) and the \emph{point configuration} is defined to be the configuration of the irreducible components of this
			curve. Each such component should have self-intersection index $-2$ and is associated with a node of an \emph{affine Dynkin diagram}; nodes
			are connected when the corresponding components intersect. The type of this Dynkin diagram is called the \emph{surface type} of the equation. 
			This description assumes that the surfaces $\mathcal{X}_{n}$ are \emph{minimal}, but it can happen that after the initial blowup procedure
			is complete, some $-1$-curves would have to be blown down. This will also result in some irreducible components having higher negative self-intersection
			index. The blowing down procedure is quite delicate, so here we assume that the surfaces $\mathcal{X}_{n}$ are indeed minimal, but see
			\cite{DzhSakTak:2013:DSTTHFADPE} and \cite{DzhKni:2019:QPEDQE} for examples requiring a blowing down.
	\item \textbf{Find a preliminary change of basis of $\operatorname{Pic}(\mathcal{X})$.} At this step, we only need to ensure that this change of basis 
			identifies the \emph{surface roots} (or nodes of the Dynkin diagrams of our surface) with the standard example. 
	\item \textbf{Find the translation vector and compare it with the standard dynamic.} Using this preliminary change of basis we can define the
			\emph{symmetry roots} for our surface that match the standard example. Using the action $\varphi_{*}$ of the mapping on $\operatorname{Pic}(\mathcal{X})$
			we can then see the induced action on the symmetry sub-lattice and, in particular, on the symmetry roots. For the discrete Painlev\'e equations, this action 
			on each root should be a translation by some multiple of the anti-canonical divisor class. Even when this vector is not the same as the translation vector 
			for the reference dynamic, it may be \emph{conjugate} to it. To find out whether this is the case, we represent each translation as a word in 
			the generators of the extended affine Weyl group and solve the conjugacy problem for words in groups. If they are conjugate, the conjugation element is 
			the necessary adjustment to our preliminary change of basis.
	\item \textbf{Find the change of variables reducing the applied problem to the standard example.} Adjusting the change of bases in $\operatorname{Pic}(\mathcal{X})$,
			if necessary, we now have the identification on the level of the Picard lattice. Next, we need to find the actual change of variables that induces
			that linear change of basis. For that, identify the curves that form the basis in the corresponding coordinate pencils. Those curves then are our 
			projective coordinates, up to a M\"obius transformation. To fix the M\"obius transformations, use the mapping of coordinate divisors. An important part
			of this computation is the identification of various parameters between the two problems. This, in fact, can be done ahead of time by using the 
			\emph{Period Map}, which gives the parameterization in terms of canonical (for the given choice of root bases) \emph{root variables}. Expressing these
			root variables in terms of parameters of the problem gives the necessary identification of parameters.
\end{enumerate}
	
	In the next section we carefully illustrate each step of this procedure using equations (\ref{eq:yn-evol}--\ref{eq:xn-evol}) as an example of an applied system. 
	Our main result is the following Theorem.
	
	\begin{theorem}\label{thm:coordinate-change}
		Recurrences (\ref{eq:yn-evol}--\ref{eq:xn-evol}) are equivalent to to the standard discrete Painlev\'e equation \eqref{eq:dPv-std}. This equivalence is achieved 
		via the following change of variables:
		\begin{equation}\label{eq:xy2fg}
			\begin{aligned}
				x(f,g)& = \gamma - g - \frac{(n + \beta)f}{f - 1},\\
				y(f,g)& = (g + \alpha + \beta + n - \gamma)(g + 2 \beta + 2n - \gamma) - n \alpha - \frac{gt (g + \beta - \gamma)}{f} \\
				&\qquad + \frac{(n + \beta)((c-1)(2g + \alpha + 3 \beta + 3n - 2 \gamma) + (\alpha + \beta + n - \gamma) + n)}{c(f-1)} + 
				\frac{(c-1)(n+ \beta)^{2}}{c(f-1)^{2}}.
			\end{aligned}
		\end{equation}
		The inverse change of variables is given by	
		\begin{equation}\label{eq:fg2xy}
			\begin{aligned}
			f(x,y) &= \frac{t(x-\beta)(x-\gamma)}{((x - \alpha) (x - \beta) - n x - y)},\\			
			g(x,y) &= -\frac{(x - \gamma)(((x - \alpha) (x - \beta) - n x - y) - t(x - \beta)(x - \gamma + \beta + n))
			}{((x - \alpha) (x - \beta) - n x - y) - t(x - \beta)(x - \gamma)}.
			\end{aligned}
		\end{equation} 
		Note that the parameters $c$ and $t$ are related by $ct = 1$.	
	\end{theorem}

	The standard difference Painlev\'e-V equation is one of the equations in the d-$\dPain{D_{4}^{(1)}/D_{4}^{(1)}}$ family of discrete Painlev\'e 
	equations whose geometric (i.e., point configuration) and algebraic (extended affine Weyl symmetry group) data are both encoded by affine Dynkin
	diagrams of type $D_{4}^{(1)}$; we collect some basic facts and data about this family in the Appendix (see also \cite{KajNouYam:2017:GAOPE}).

\section{The Identification Procedure} 
\label{sec:the_identification_procedure}

\subsection{The Singularity Structure} 
\label{sub:the_singularity_structure}

The first step in the geometric analysis of discrete Painlev\'e equations is to understand the singularity structure of the system (which is best done
with the help of some Computer Algebra System; in this project we used \textbf{Mathematica}${}^{\text{\circledR}}$).

Note that equation~\eqref{eq:yn-evol} defines the \emph{forward} mapping $\psi_{1}^{(n)}:(x_{n},y_{n})\mapsto (x_{n}, y_{n+1})$ and equation~\eqref{eq:xn-evol} defines
the \emph{backward} mapping $\psi_{2}^{(n)}: (x_{n},y_{n}) \mapsto (x_{n-1},y_{n})$, this is fairly typical for discrete Painlev\'e equations obtained as
deautonomizations of QRT mappings, see \cite{CarDzhTak:2017:FDOI2MADPE}.

First, we compactify the affine complex plane $\mathbb{C}^{2}$ to $\mathbb{P}^{1} \times \mathbb{P}^{1}$ by introducing homogeneous coordinates
$[x^{0}:x^{1}]$ and $[y^{0}:y^{1}]$ with $x = \frac{x^{0}}{x^{1}}$ in the affine chart $x_{1}\neq 0$, $X = \frac{1}{x} = \frac{x^{1}}{x^{0}}$ in the
affine chart $x_{0}\neq 0$ and $y$ and $Y = \frac{1}{y}$ defined similarly. Next, we look for indeterminacies of rational maps, i.e., the points where
both the numerator and the denominator of the map vanish. At those points we perform the blowup procedure, see, e.g., \cite{Sha:2013:BAG1}, 
which essentially is an introduction of two new charts $(u_{i},v_{i})$ and $(U_{i},V_{i})$ in the neighborhood of the blowup point $q_{i}(x_{i},y_{i})$, where the change of 
variables is given by $x = x_{i} + u_{i} = x_{i} + U_{i} V_{i}$ and $y = y_{i} + u_{i} v_{i} =y_{i} +  V_{i}$. The coordinates $v_{i} = 1/U_{i}$ represent all
possible slopes of lines passing through the point $q_{i}$, and so this variable change ``separates'' all curves passing through $q_{i}$ based on their
slopes. This change of variables is a bijection away from $q_{i}$, but the point $q_{i}$ is replaced by the $\mathbb{P}^{1}$-line of all possible slopes, 
called the \emph{central fiber} or the \emph{exceptional divisor} of the blowup. We denote this central fiber by $F_{i}$, it is given in the blowup charts 
by local equations $u_{i}=0$ and $V_{i} = 0$.

\subsubsection{The Forward Mapping} 
\label{ssub:the_forward_mapping}

We begin by considering the forward mapping. We put $\overline{x}=x:=x_{n}$, 
$y:=y_{n}$, $\overline{y}: = y_{n+1}$ and omit the index $n$
in the mapping notation. The map $\psi_{1}:(x,y)\mapsto (\overline{x},\overline{y})$ then becomes
\begin{equation}\label{eq:fwd}
	(\overline{x},\overline{y}) = \left(x, \frac{(x-1)(x-\alpha)(x-\beta)(x-\gamma)}{c (y - (x - \alpha)(x - \beta)+n x)} + (x - \alpha)(x - \beta)-(n+1)x\right),
\end{equation}
and we immediately see the following base points (in the affine coordinates $(x,y)$):
\begin{equation*}
	q_{1}(1,(1 - \alpha)(1-\beta) - n),\quad q_{2}(\alpha,-n \alpha), \quad q_{3}(\beta, - n \beta),\quad q_{4}(\gamma, (\gamma - \alpha)(\gamma - \beta) - n \gamma).
\end{equation*}
Rewriting the mapping for $\overline{y}$ in the $(X,Y)$-chart, we get
\begin{equation*}
	\overline{y} = \frac{\left(\begin{aligned}
		& Y(1 - X)(1 - \alpha X)(1 - \beta X)(1 - \gamma X) \\
		&\qquad\qquad\qquad + c \Big(X^{2}-Y\left((1 - \alpha X)(1 - \beta X) - n X\right)\Big)\Big((1 - \alpha X)(1 - \beta X)-(n+1)X\Big)
		\end{aligned}\right)
	}{
	c X^{2}\Big(X^{2}-Y\left((1 - \alpha X)(1 - \beta X) - n X \right)\Big)
	},
\end{equation*}
we see that we get a new base point $q_{5}(x=\infty,y=\infty)$ or $q_{5}(X=0,Y=0)$.  It is easy to see that these points are the only
base points on $\mathbb{P}^{1} \times \mathbb{P}^{1}$ for the forward dynamic. Thus, if this mapping
is indeed in the discrete Painlev\'e family, there are three more points on exceptional divisors (these points can also appear for
the backward dynamic, but we show later that this is not the case).

\paragraph{Resolving $q_{1},\ldots,q_{4}$.} 
\label{par:resolving_p__1_ldots_p__4}
We introduce blowup coordinates at $q_{1}(1,(1 - \alpha)(1-\beta) - n)$ via
\begin{equation*}
	x = 1 + u_{1} = 1 + U_{1} V_{1},\qquad y = (1 - \alpha)(1-\beta) - n + u_{1} v_{1} = (1 - \alpha)(1-\beta) - n + V_{1}.
\end{equation*}
In the coordinates $(u_{1},v_{1})$ we get
\begin{equation}\label{eq:mapF1f}
	\overline{x} = 1 + u_{1},\quad \overline{y} = \frac{
	u_{1} (1 - \alpha + u_{1} )(1 - \beta + u_{1} )(1  - \gamma + u_{1})
	}{
	c u_{1} (v_{1} - (2 - \alpha - \beta + u_{1}) + n)
	} + (1 - \alpha + u_{1} )(1 - \beta + u_{1} ) - (n+1)(1 + u_{1}),
\end{equation}
and we first see that the cancelation of $u_{1}$ in the fraction resolves the indeterminacy, and so the mapping lifts 
to the exceptional divisor $F_{1}$ whose equation in this chart is $u_{1} = 0$.
%
Studying the mapping in the $(U_{1},V_{1})$-chart does not give any new information.

The computation is exactly the same for the points $q_{2},\ldots,q_{4}$, the mapping extends without
new base points to the exceptional divisors $F_{i}$.


\paragraph{Resolving $q_{5}$ and its degeneration cascade.} 
\label{par:resolving_p__5_and_its_degeneration_cascade}

The situation at the point $q_{5}(\infty,\infty)$ is more interesting. Introducing blowup coordinates at this point via
\begin{equation*}
	X = u_{5} = U_{5} V_{5},\qquad Y = u_{5} v_{5} = V_{5},
\end{equation*}
and considering the mapping in the $(u_{5},v_{5})$-chart, we get, after cancelling the $u_{5}$-factor in the numerator and denominator,
\begin{equation*}
	\overline{y}(u_{5},v_{5}) =
	\frac{\left(\begin{aligned}
			& v_{5}(1 - u_{5})(1 - \alpha u_{5})(1 - \beta u_{5})(1 - \gamma u_{5}) \\
			&\qquad\qquad\qquad + c \Big(u_{5}-v_{5}\left((1 - \alpha u_{5})(1 - \beta u_{5}) - n u_{5}\right)\Big)\Big((1 - \alpha u_{5})(1 - \beta u_{5})-(n+1)u_{5}\Big)
			\end{aligned}\right)
		}{
		c u_{5}^{2}\Big(u_{5} - v_{5}\left((1 - \alpha u_{5})(1 - \beta u_{5}) - n u_{5}\right)\Big)
		}.
\end{equation*}
We see that this mapping has a new base point $q_{6}(u_{5} = X = 0, v_{5} = Y/X = 0)$ (note that this base point is not visible in the $(U_{5},V_{5})$-chart).
Continuing in this way, we get the following cascade of ``infinitely close'' base points:
\begin{align*}
	q_{5}(X=0,Y=0)&\leftarrow q_{6}\left(u_{5} = X = 0, v_{5} = \frac{Y}{X} = 0\right) \leftarrow
	q_{7}\left(u_{6} = u_{5} = 0, v_{6} = \frac{v_{5}}{u_{5}}= \frac{c}{c-1}\right) \\
	&\leftarrow  q_{8}\left(u_{7} = u_{6} = 0, v_{7} = \frac{(c - 1)v_{6} - c}{(c-1)u_{6}} = \frac{c\Big(c(\alpha + \beta + n) + n - \gamma\Big)}{(c-1)^{2}}\right).
\end{align*}


Note that the positions of base points depend on $n$, and so evolve with the dynamics, but the \emph{configuration} of base points remains fixed. 
Put $\mathcal{X}_{n}:= \operatorname{Bl}_{q_{1}\cdots q_{8}}(\mathbb{P}^{1} \times \mathbb{P}^{1})$ and let 
$\eta_{n}: \mathcal{X}_{n}\to \mathbb{P}^{1} \times \mathbb{P}^{1}$
be the corresponding blow down map. This gives us a typical surface in the family on which the dynamic is defined. In what follows we may sometimes omit the 
index $n$ when only the point configuration and not the exact location of the base points is important.

\subsubsection{The Backward Mapping} 
\label{ssub:the_backward_mapping}

Consider now the backward mapping. We put $x:=x_{n}$, $\underline{x} = x_{n-1}$, $\underline{y} = y:=y_{n}$, 
The backward mapping $\psi_{2}:(x,y)\mapsto (\underline{x},\underline{y})$ then becomes
\begin{equation}\label{eq:back}
	(\underline{x},\underline{y}) = \left(
	\frac{
		(y + n \alpha)(y +n \beta)(y + n \gamma - (\gamma - \alpha)(\gamma - \beta))(y - n - (1 - \alpha)(1 - \beta))
		}{
		(x + \mathfrak{Y})\left( y (2n + \alpha + \beta - \gamma- 1) +
		n\left( (n + \alpha + \beta)(n + \alpha + \beta - \gamma - 1) - \alpha \beta + \gamma\right)\right)^{2}
		} - \mathfrak{Y},
	y\right),
\end{equation}
where $\mathfrak{Y}$ is given by \eqref{eq:Yn} (we omit the index $n$). 
The same standard computation shows that the only base points of the backwards dynamic are
the same points $q_{1},\ldots,q_{4}$ as for the forward dynamic,
but the singularity cascade at $q_{5}$ is not present.



\subsection{The Mapping on $\operatorname{Pic}(\mathcal{X})$} 
\label{sub:the_mapping_on_operatorname_pic_mathcal_x}

Recall that for a regular algebraic variety $\mathcal{X}$, its \emph{Picard group} (or \emph{Picard lattice}) is the quotient 
of the \emph{divisor group} $\operatorname{Div}(\mathcal{X})=\operatorname{Span}_{\mathbb{Z}}(\mathcal{D})$, that is a free Abelian
group generated by closed irreducible subvarieties $D$ of codimension $1$, by the subgroup $\operatorname{P}(\mathcal{X})$ of 
\emph{principal divisors} (i.e., by the relation of \emph{linear equivalence}),
\begin{equation*}
	\operatorname{Pic}(\mathcal{X}) \simeq \operatorname{Cl}(\mathcal{X}) = \operatorname{Div}(\mathcal{X})/\operatorname{P}(\mathcal{X}) 
	= \operatorname{Div}(\mathcal{X})/\sim,
\end{equation*}
see \cite{SmiKahKek:2000:IAG} or \cite{Sha:2013:BAG1}. In our case, it is enough to know that 
$\operatorname{Pic}(\mathbb{P}^{1} \times \mathbb{P}^{1}) = \operatorname{Span}_{\mathbb{Z}}\{\mathcal{H}_{x},\mathcal{H}_{y}\}$, where 
$\mathcal{H}_{x} = [H_{x=a}]$ is the class of a \emph{vertical}  and  $\mathcal{H}_{y} = [H_{y=b}]$ is the class of a \emph{horizontal} 
line on $\mathbb{P}^{1}\times \mathbb{P}^{1}$. Each blowup procedure at a point $q_{i}$ adds the class $\mathcal{F}_{i} = [F_{i}]$
of the \emph{exceptional divisor} (i.e., the \emph{central fiber}) of the blowup, so 
$\operatorname{Pic}(\mathcal{X}_{n}) = \operatorname{Span}_{\mathbb{Z}}\{\mathcal{H}_{x},\mathcal{H}_{y},\mathcal{F}_{1},\ldots, \mathcal{F}_{8}\}$.
Further, the Picard lattice is equipped with the symmetric bilinear \emph{intersection form} given by 
\begin{equation}\label{eq:int-form}
\mathcal{H}_{x}\bullet \mathcal{H}_{x} = \mathcal{H}_{y}\bullet \mathcal{H}_{y} = \mathcal{H}_{x}\bullet \mathcal{E}_{i} = 
\mathcal{H}_{y}\bullet \mathcal{E}_{j} = 0,\qquad \mathcal{H}_{x}\bullet \mathcal{H}_{y} = 1,\qquad  \mathcal{E}_{i}\bullet \mathcal{E}_{j} = - \delta_{ij}	
\end{equation}
on the generators, and then extended by linearity.

Both the forward and the backward mappings induce linear maps on $\operatorname{Pic}(\mathcal{X})$. We use the notation 
$\operatorname{Pic}(\overline{\mathcal{X}})$ (resp.~$\operatorname{Pic}(\underline{\mathcal{X}})$) for the range of the 
forward (resp.~backward) mappings;
note that all these groups are clearly canonically isomorphic, so we sometimes just use the notation 
$\operatorname{Pic}(\mathcal{X})$. We use $\overline{F}_{i}$ to denote the divisor of the central fiber of the blowup at the point 
$\overline{q}_{i} = \psi_{1}(q_{i})$, and similarly for the backwards mapping and for the classes.

Since the mapping is not very complicated, we can compute its action on $\operatorname{Pic}(\mathcal{X})$ directly. The result is given by the
following Lemma, where we use the notation $\mathcal{F}_{ij} = \mathcal{F}_{i} + \mathcal{F}_{j}$ and so on.

\begin{lemma}\label{lem:dyn}
	\qquad

	\begin{enumerate}[(a)]
	\item
	The action of the forward dynamic $(\psi_{1})_{*}: \operatorname{Pic}(\mathcal{X})\to \operatorname{Pic}(\overline{\mathcal{X}})$
	is given by
	\begin{alignat*}{5}
		\mathcal{H}_{x}& \mapsto \overline{\mathcal{H}}_{x}, \quad&
		\mathcal{F}_{1} &\mapsto \overline{\mathcal{H}}_{x} -\overline{\mathcal{F}}_{1}, \quad&
		\mathcal{F}_{3} &\mapsto \overline{\mathcal{H}}_{x} -\overline{\mathcal{F}}_{3}, \quad&
		\mathcal{F}_{5} &\mapsto \overline{\mathcal{H}}_{x} -\overline{\mathcal{F}}_{8}, \quad&
		\mathcal{F}_{7} &\mapsto \overline{\mathcal{H}}_{x} -\overline{\mathcal{F}}_{6}, \\
		\mathcal{H}_{y}& \mapsto 4\overline{\mathcal{H}}_{x} + \overline{\mathcal{H}}_{y} -\overline{\mathcal{F}}_{12345678},  \quad&
		\mathcal{F}_{2} &\mapsto \overline{\mathcal{H}}_{x} -\overline{\mathcal{F}}_{2}, \quad&
		\mathcal{F}_{4} &\mapsto \overline{\mathcal{H}}_{x} -\overline{\mathcal{F}}_{4}, \quad&
		\mathcal{F}_{6} &\mapsto \overline{\mathcal{H}}_{x} -\overline{\mathcal{F}}_{7}, \quad&
		\mathcal{F}_{8} &\mapsto \overline{\mathcal{H}}_{x} -\overline{\mathcal{F}}_{5},
	\end{alignat*}
	and  the evolution of base points $\overline{q}_{i} = \psi_{1}(q_{i})$ is given by
	\begin{equation*}
		\overline{q}_{1}(1,(1 - \alpha)(1 - \beta)-(n+1)),\  \overline{q}_{2}(\alpha,-(n+1)\alpha), \
		\overline{q}_{3}(\beta,-(n+1)\beta), \  \overline{q}_{4}(\gamma,(\gamma - \alpha)(\gamma - \beta)-(n+1)\gamma)
	\end{equation*}
	for finite points, and for the degeneration cascade we get
	\begin{align*}
	\overline{q}_{5}(\overline{X}=0,\overline{Y}=0)&\leftarrow \overline{q}_{6}\left(\overline{u}_{5} = \overline{X} = 0, \overline{v}_{5}
	= \frac{\overline{Y}}{\overline{X}} = 0\right) \leftarrow
	\overline{q}_{7}\left(\overline{u}_{6} = \overline{u}_{5} = 0, \overline{v}_{6} = \frac{\overline{v}_{5}}{\overline{u}_{5}}= \frac{c}{c-1}\right) \\
	&\leftarrow  \overline{q}_{8}\left(\overline{u}_{7} = \overline{u}_{6} = 0,
	\overline{v}_{7} = \frac{(c - 1)\overline{v}_{6} - c}{(c-1)\overline{u}_{6}} =
	\frac{c\Big(c(\alpha + \beta + n + 1) + n - \gamma - 1\Big)}{(c-1)^{2}}\right).
	\end{align*}

	\item
	The action of the backwards dynamic $(\psi_{2})_{*}: \operatorname{Pic}(\mathcal{X})\to \operatorname{Pic}(\underline{\mathcal{X}})$
	is given by
	\begin{alignat*}{5}
		\mathcal{H}_{x}& \mapsto \underline{\mathcal{H}}_{x} + 2 \underline{\mathcal{H}}_{y} - \underline{\mathcal{F}}_{1234}, \quad&
		\mathcal{F}_{1} &\mapsto \underline{\mathcal{H}}_{y} -\underline{\mathcal{F}}_{1}, \quad&
		\mathcal{F}_{3} &\mapsto \underline{\mathcal{H}}_{y} -\underline{\mathcal{F}}_{3}, \quad&
		\mathcal{F}_{5} &\mapsto \underline{\mathcal{F}}_{5}, \quad&
		\mathcal{F}_{7} &\mapsto \underline{\mathcal{F}}_{7}, \\
		\mathcal{H}_{y}& \mapsto \underline{\mathcal{H}}_{y},  \quad&
		\mathcal{F}_{2} &\mapsto \underline{\mathcal{H}}_{y} -\underline{\mathcal{F}}_{2}, \quad&
		\mathcal{F}_{4} &\mapsto \underline{\mathcal{H}}_{y} -\underline{\mathcal{F}}_{4}, \quad&
		\mathcal{F}_{6} &\mapsto \underline{\mathcal{F}}_{6}, \quad&
		\mathcal{F}_{8} &\mapsto \underline{\mathcal{F}}_{8}.
	\end{alignat*}
	From this we can also easily compute the evolution of base points. We get
	\begin{equation*}
		\underline{q}_{1}(1,(1 - \alpha)(1 - \beta)-n),\  \underline{q}_{2}(\alpha,-n\alpha), \
		\underline{q}_{3}(\beta,-n\beta), \  \underline{q}_{4}(\gamma,(\gamma - \alpha)(\gamma - \beta)-n\gamma),
	\end{equation*}
	as well as the degeneration cascade
	\begin{align*}
	\underline{q}_{5}(\underline{X}=0,\underline{Y}=0)&\leftarrow \underline{q}_{6}\left(\underline{u}_{5} = \underline{X} = 0, \underline{v}_{5}
	= \frac{\underline{Y}}{\underline{X}} = 0\right) \leftarrow
	\underline{q}_{7}\left(\underline{u}_{6} = \underline{u}_{5} = 0, \underline{v}_{6} = \frac{\underline{v}_{5}}{\underline{u}_{5}}= \frac{c}{c-1}\right) \\
	&\leftarrow  \underline{q}_{8}\left(\underline{u}_{7} = \underline{u}_{6} = 0,
	\underline{v}_{7} = \frac{(c - 1)\underline{v}_{6} - c}{(c-1)\underline{u}_{6}} =
	\frac{c\Big(c(\alpha + \beta + n) + n - \gamma - 2\Big)}{(c-1)^{2}}\right).
	\end{align*}
	\item
	The action of the composed mapping $\psi^{(n)}_{*} = \psi_{*} = (\psi_{2})_{*}^{-1}\circ (\psi_{1})_{*}: 
	\operatorname{Pic}(\mathcal{X}_{n})\to \operatorname{Pic}(\mathcal{X}_{n+1})$
	is given by
	\begin{alignat*}{2}
		\mathcal{H}_{x}& \mapsto \overline{\mathcal{H}}_{x} + 2 \overline{\mathcal{H}}_{y} - \overline{\mathcal{F}}_{1234}, \qquad&
		\mathcal{H}_{y}& \mapsto 4\overline{\mathcal{H}}_{x} + 5 \overline{\mathcal{H}}_{y} - 3\overline{\mathcal{F}}_{1234}-\overline{\mathcal{F}}_{5678},  \\
		\mathcal{F}_{1} &\mapsto \overline{\mathcal{H}}_{x} + \overline{\mathcal{H}}_{y} -\overline{\mathcal{F}}_{234}, \qquad&
		\mathcal{F}_{5} &\mapsto \overline{\mathcal{H}}_{x} + 2\overline{\mathcal{H}}_{y} -\overline{\mathcal{F}}_{12348}\\
		\mathcal{F}_{2} &\mapsto \overline{\mathcal{H}}_{x} + \overline{\mathcal{H}}_{y} -\overline{\mathcal{F}}_{134},\qquad&
		\mathcal{F}_{6} &\mapsto \overline{\mathcal{H}}_{x} + 2\overline{\mathcal{H}}_{y} -\overline{\mathcal{F}}_{12347}, \\		
		\mathcal{F}_{3} &\mapsto \overline{\mathcal{H}}_{x} + \overline{\mathcal{H}}_{y} -\overline{\mathcal{F}}_{124}, \qquad&
		\mathcal{F}_{7} &\mapsto \overline{\mathcal{H}}_{x} + 2\overline{\mathcal{H}}_{y} -\overline{\mathcal{F}}_{12346}, \\
		\mathcal{F}_{4} &\mapsto \overline{\mathcal{H}}_{x} + \overline{\mathcal{H}}_{y} -\overline{\mathcal{F}}_{123}, \qquad&
		\mathcal{F}_{8} &\mapsto \overline{\mathcal{H}}_{x} + 2\overline{\mathcal{H}}_{y} -\overline{\mathcal{F}}_{12345}.
	\end{alignat*}
	The evolution of the base points (here $\overline{q}_{i} = \psi^{(n)}(q_{i})$) is
	\begin{equation*}
		\overline{q}_{1}(1,(1 - \alpha)(1 - \beta)-(n+1)),\  \overline{q}_{2}(\alpha,-(n+1)\alpha), \
		\overline{q}_{3}(\beta,-(n+1)\beta), \  \overline{q}_{4}(\gamma,(\gamma - \alpha)(\gamma - \beta)-(n+1)\gamma)
	\end{equation*}
	for finite points, and 
	\begin{align*}
	\overline{q}_{5}(\overline{X}=0,\overline{Y}=0)&\leftarrow \overline{q}_{6}\left(\overline{u}_{5} = \overline{X} = 0, \overline{v}_{5}
	= \frac{\overline{Y}}{\overline{X}} = 0\right) \leftarrow
	\overline{q}_{7}\left(\overline{u}_{6} = \overline{u}_{5} = 0, \overline{v}_{6} = \frac{\overline{v}_{5}}{\overline{u}_{5}}= \frac{c}{c-1}\right) \\
	&\leftarrow  \overline{q}_{8}\left(\overline{u}_{7} = \overline{u}_{6} = 0,
	\overline{v}_{7} = \frac{(c - 1)\overline{v}_{6} - c}{(c-1)\overline{u}_{6}} =
	\frac{c\Big(c(\alpha + \beta + n+1) + n+1 - \gamma\Big)}{(c-1)^{2}}\right)
	\end{align*}
	for the degeneration cascade.

	\end{enumerate}
Moreover, from the evolution of base points we see that $\psi_{1}(n) = n+1$, $\psi_{2}(n) = n$, and $\psi^{(n)}(n) = n+1$.
\end{lemma}

\begin{proof} This is a standard computation in charts that we illustrate by a few examples for the forward mapping $\psi_{1}$.
	First, since $\psi_{1}(x) = x$, we see that $(\psi_{1})_{*} (\mathcal{H}_{x}) = \overline{\mathcal{H}}_{x}$. 	
	To find $(\psi_{1})_{*} (F_{1})$ we restrict the mapping \eqref{eq:mapF1f} to $u_{1} = 0$ to get
	\begin{equation*}
		\psi_{1}(u_{1}=0,v_{1}) = \left(1,\frac{(1-\alpha)(1-\beta)(1-\gamma)}{c(v_{1} + \alpha + \beta + n - 2)} + 
		(1 - \alpha)(1-\beta) - (n+1)\right), 
	\end{equation*}
	so $\bar{y}$ is just a fractional linear transformation of the parameter $v_{1}$ and thus the exceptional divisor $F_{1}$ maps, 
	parametrically, to the line $\overline{x}=1$ in the affine chart $(\overline{x},\overline{y})$ of the 
	range. This line lifts, as a proper transform, to the divisor $\overline{H}_{x} - \overline{F}_{1}$ on the 
	surface $\overline{\mathcal{X}}$, which, after passing to divisor classes, 
	gives the required map on $\operatorname{Pic}(\mathcal{X})$. Note that this also implies that 
	$(\psi_{1})_{*}(H_{x}-F_{1}) = \overline{F}_{1}$, i.e., $\psi_{1}(x=1) = (1,(1-\alpha)(1-\beta) - (n+1)) = \overline{q}_{1}$, and so 
	we see that under $\psi_{1}$ the step $n$ evolves to $n+1$. The computations for $F_{2},\ldots,F_{4}$ are very similar.
	
	This computation gets slightly more complicated in the degeneration cascade. For example, 
	to find $(\psi_{1})_{*}(F_{5})$ we need to compute $\psi_{1}$ in the chart $(u_{5},v_{5})$ and restrict to $u_{5}=0$. However, since 
	there is a base point $q_{6}$ on $F_{5}$, the mapping $\psi_{1}(u_{5}=0,v_{5})$ corresponds to mapping, parametrically, the \emph{proper transform}
	$F_{5} - F_{6}$ of $F_{5}$ on $\mathcal{X}$. We compute $\psi_{1}(0,v_{5}) = (\infty,\infty)$, so $F_{5} - F_{6}$ collapses onto the base point 
	$\overline{q}_{5}$ on $\mathbb{P}^{1} \times \mathbb{P}^{1}$. Switching to coordinates $(\overline{u}_{5},\overline{v}_{5})$ in the range, we get
	$(\overline{u}_{5},\overline{v}_{5})(0,v_{5}) = (0,0)$. Thus, we have further collapse to $\overline{q}_{6}$ and need to do the computation 
	in the chart $(\overline{u}_{6},\overline{v}_{6})$. We get $(\overline{u}_{6},\overline{v}_{6})(0,v_{5}) = \overline{q}_{7}$ and finally,
	\begin{equation*}
		(\overline{u}_{7},\overline{v}_{7})(0,v_{5}) = 
		\left(0, \frac{c + c (-1 + n + c (1 + n + \alpha + \beta) -\gamma)v_{5}}{(c-1)^{2}v_{5}}\right),
	\end{equation*}
	and so $(\psi_{1})_{*} (F_{5} - F_{6}) = \overline{F}_{7} - \overline{F}_{8}$. Note that the mappings $(\psi_{i})_{*}$ should preserve the 
	intersection form, and thus the self-intersection index. Indeed $(F_{5} - F_{6})^{2} = (\overline{F}_{7} - \overline{F}_{8})^{2} = -2$, as it should be.
	Other computations in the degeneration cascade are similar and result in
	\begin{equation*}
		(\psi_{1})_{*} (F_{6} - F_{7}) = \overline{F}_{6} - \overline{F}_{7},\quad 
		(\psi_{1})_{*} (F_{7} - F_{8}) = \overline{F}_{5} - \overline{F}_{6},\quad 
		(\psi_{1})_{*} (F_{8}) = \overline{H}_{x} - \overline{F}_{5}.
	\end{equation*}
	Passing to classes, we get
	\begin{equation*}
		(\psi_{1})_{*} (\mathcal{F}_{7} )  = (\psi_{1})_{*} (\mathcal{F}_{7} - \mathcal{F}_{8}) + (\psi_{1})_{*} (\mathcal{F}_{8}) 
		= (\overline{\mathcal{F}}_{5} - \overline{\mathcal{F}}_{6}) + (\overline{\mathcal{H}}_{x} - 
		\overline{\mathcal{F}}_{5}) = \overline{\mathcal{H}}_{x} - \overline{\mathcal{F}}_{6},
	\end{equation*}
	and so on.
	
	To find $(\psi_{1})_{*}(\mathcal{H}_{y})$, it is  convenient to choose a vertical line passing through some base points. For example
	$H_{y} - F_{5} - F_{6}$ is the proper transform of the line $y = \infty$, and we get
	\begin{equation*}
		(\overline{x},\overline{y})(Y=0) = (x,x^{2} + \alpha \beta - x(n+1 + \alpha + \beta)),		
	\end{equation*}
	which parameterizes the $(2,1)$-curve $\overline{\gamma}$ in the range given by the equation 
	$\overline{x}^{2} + \alpha \beta - \overline{x}(n+1 + \alpha + \beta) - \overline{y}=0$. 
	Taking into account the evolution of parameters we note that this curve corresponds to the $(2,1)$-curve $\gamma$ in the domain given by
	\begin{equation}\label{eq:gamma-curve}
		\gamma:\quad 	x^{2} + \alpha \beta - x(n + \alpha + \beta) - y=0.
	\end{equation}
	Since the self-intersection index of $\overline{\gamma}$ should be $-2$, we expect it to pass through $6$ (or fewer, in case of multiplicities) 
	base points in the range. Indeed, we can check that it passes through $\overline{q}_{1},\ldots,\overline{q}_{6}$, each with multiplicity one, and so 
	\begin{equation*}
		(\psi_{1})_{*}(\mathcal{H}_{y}) =  
		(\psi_{1})_{*}(\mathcal{H}_{y} - \mathcal{F}_{5} - \mathcal{F}_{6}) +  
		(\psi_{1})_{*}(\mathcal{F}_{5}) + (\psi_{1})_{*}(\mathcal{F}_{6})= (2 \overline{\mathcal{H}}_{x} + \overline{\mathcal{H}}_{y}
		- \overline{F}_{123456}) + (\overline{\mathcal{H}}_{x} - \overline{\mathcal{F}}_{8}) + (\overline{\mathcal{H}}_{x} - \overline{F}_{7}).
	\end{equation*}
	This completes the computation for the forward mapping $\psi_{1}$. The computation for the backward mapping $\psi_{2}$ is similar, and the computation
	for the composed mapping $\psi^{(n)}$ immediately follows.
\end{proof}

\subsection{The Surface Type} 
\label{sub:the_surface_type}

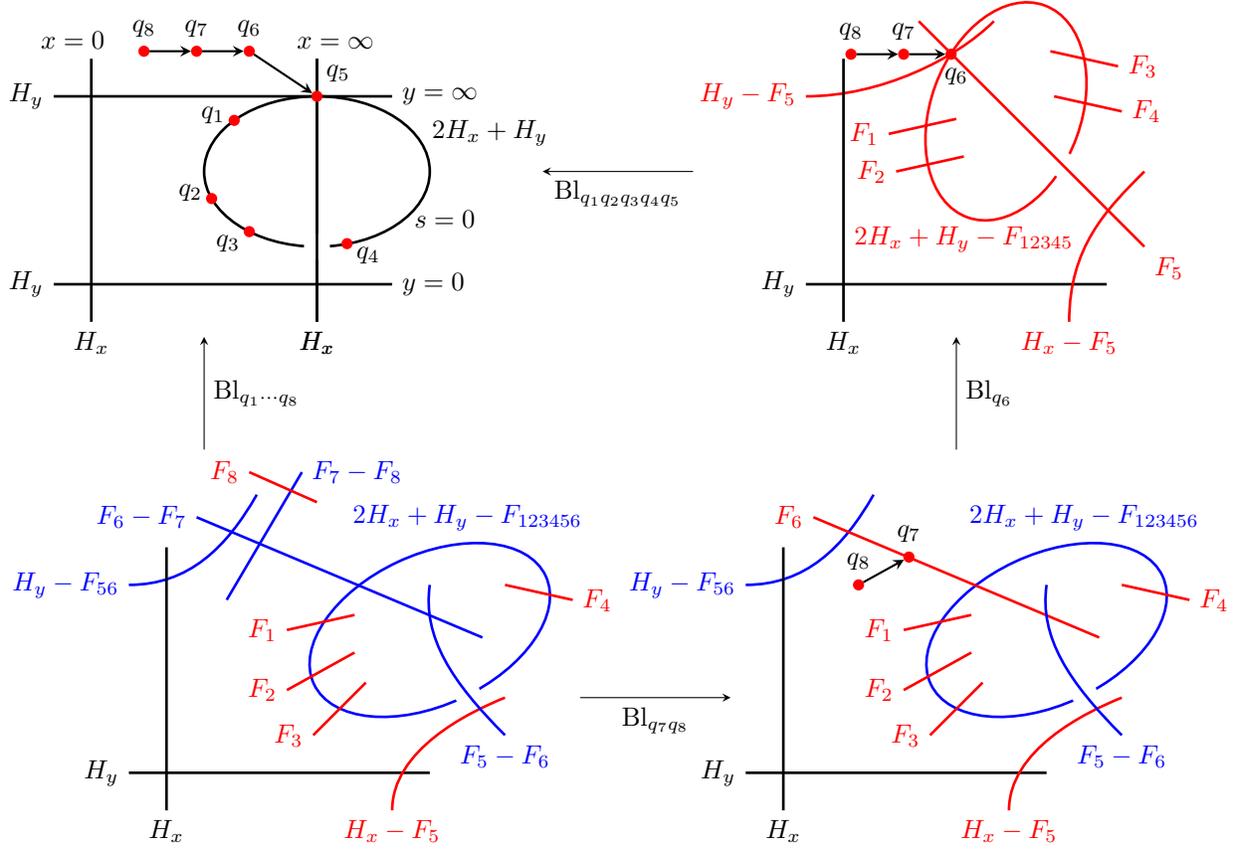
\begin{figure}[ht]
	\begin{tikzpicture}[>=stealth,basept/.style={circle, draw=red!100, fill=red!100, thick, inner sep=0pt,minimum size=1.2mm}]
		\begin{scope}[xshift = -1cm]
			\draw [black, line width = 1pt]  	(4,0) 	-- (-0.5,0) 		node [left] {$H_{y}$} node[pos=0, right] {$y=0$};
			\draw [black, line width = 1pt] 	(4,2.5) 	-- (-0.5,2.5)	node [left] {$H_{y}$} node[pos=0, right] {$y=\infty$};
			\draw [black, line width = 1pt] 	(0,3) -- (0,-0.5)			node [below] {$H_{x}$} node[pos=0, above, xshift=-7pt] {$x=0$};
			\draw [black, line width = 1pt] 	(3,3) -- (3,-0.5)			node [below] {$H_{x}$} node[pos=0, above, xshift=7pt] {$x=\infty$};

			\draw[line width = 1pt] (3,1.5) circle [x radius=1.5cm, y radius = 1cm];
			\node at (3,0.5) 	[circle, draw=white!100, fill=white!100, minimum size=1mm] {};
			\draw [black, line width = 1pt] 	(3,3) -- (3,-0.5)			node [below] {$H_{x}$};

			\node (p1) at (1.9,2.18) [basept,label={[left] $q_{1}$}] {};
			\node (p2) at (1.6,1.14) [basept,label={[left] $q_{2}$}] {};
			\node (p3) at (2.1,0.7) [basept,label={[below left] $q_{3}$}] {};
			\node (p4) at (3.4,0.54) [basept,label={[below right] $q_{4}$}] {};
			\node (p5) at (3,2.5) 	[basept,label={[above right] $q_{5}$}] {};
			\node (p6) at (2.1,3.1)	[basept,label={[above] 		$q_{6}$}] {};
			\node (p7) at (1.4,3.1) [basept,label={[above] 		$q_{7}$}] {};
			\node (p8) at (0.7,3.1) [basept,label={[above] 		$q_{8}$}] {};

			\node at (4.7,0.5) [label = {$s = 0$}] {};
			\node at (5.3,1.6) [label = {$2 H_{x} + H_{y}$}] {};

			\draw [line width = 0.8pt, ->]  (p8) edge (p7) (p7) edge (p6) (p6) edge (p5);
		\end{scope}

		\draw [->] (7,1.5)--(5,1.5) node[pos=0.5, below] {$\operatorname{Bl}_{q_{1}q_{2}q_{3}q_{4}q_{5}}$};

		\begin{scope}[xshift = 9cm]
			\draw [black, line width = 1pt ] (3.5,0) -- (-0.5,0) node [left] {$H_{y}$};
			\draw [red, line width = 1pt ] (2,3.5) .. controls (1.5,3) and (0.5,2.5) .. (-0.5,2.5) node [left] {$H_{y} - F_{5}$};
			\draw [black, line width = 1pt ] (0,3) -- (0,-0.5) node [below] {$H_{x}$};
			\draw [red, line width = 1pt ] (4,1.5) .. controls (3.5,1) and (3,0.5) .. (3,-0.5) node [below] {$H_{x} - F_{5}$};

			\draw[red, line width = 1pt] (2.16,2.3) circle [x radius=1.5cm, y radius = 1cm, rotate = 70];
			\node at (2.91,1.59) 	[circle, draw=white!100, fill=white!100, minimum size=0.8mm] {};

			\node at (1.6,0.6) [red] {$2H_{x} + H_{y} - F_{12345} $};

			\draw [red, line width = 1 pt] (1,3.5) -- (4,0.5)node [below right] {$F_{5}$};

			\draw [red, line width = 1 pt] (0.6,2) -- (1.5,2.2) node [pos = 0, left] {$F_{1}$};
			\draw [red, line width = 1 pt] (0.7,1.5) -- (1.6,1.7) node [pos = 0, left] {$F_{2}$};
			\draw [red, line width = 1 pt] (2.75,3.1) -- (3.65,2.9) node [right]{$F_{3}$};
			\draw [red, line width = 1 pt] (2.8,2.5) -- (3.7,2.3) node [right] {$F_{4}$};
			\node (p6) at (1.43,3.06) [basept, label={[xshift=2pt, yshift=-17pt] $q_{6}$}] {};
			\node (p7) at (0.8,3.06) [basept,label={[above] 		$q_{7}$}] {};
			\node (p8) at (0.1,3.06) [basept,label={[above] 		$q_{8}$}] {};

			\draw [line width = 0.8pt, ->]  (p8) edge (p7) (p7) edge (p6);
		\end{scope}

		\draw [->] (10.5,-2.2)--(10.5,-0.7) node[pos=0.5, right] {$\operatorname{Bl}_{q_{6}}$};

		\begin{scope}[xshift = 8.2cm, yshift = -6.5cm]
			\draw [black, line width = 1pt ] (3.5,0) -- (-0.5,0) node [left] {$H_{y}$};
			\draw [blue, line width = 1pt ] (1.2,3.7) .. controls (0.8,3) and (0.3,2.5) .. (-0.5,2.5) node [left] {$H_{y} - F_{56}$};
			\draw [black, line width = 1pt ] (0,3) -- (0,-0.5) node [below] {$H_{x}$};
			\draw [red, line width = 1pt ] (4.5,1) .. controls (4,0.8) and (3,0.3) .. (3,-0.5) node [below] {$H_{x} - F_{5}$};

			\draw[blue, line width = 1pt] (3.5,1.9) circle [x radius=1.7cm, y radius = 1cm, rotate = 25] ;
			\node at (4,1.05) 	[circle, draw=white!100, fill=white!100, minimum size=1mm] {};

			\node at (4,3.4) [blue] {$2H_{x} + H_{y} - F_{123456} $};

			\draw [red, line width = 1 pt] (1.6,1.9) -- (2.5,2.1) node [pos = 0, left] {$F_{1}$};
			\draw [red, line width = 1 pt] (1.6,1.1) -- (2.5,1.6) node [pos = 0, left] {$F_{2}$};
			\draw [red, line width = 1 pt] (1.95,0.5) -- (2.65,1.2) node [pos = 0, left] {$F_{3}$};
			\draw [red, line width = 1 pt] (4.5,2.5) -- (5.4,2.3) node [right] {$F_{4}$};

			\draw [red, line width = 1 pt] (0.4,3.4) -- (4.2,1.8) node [pos = 0, left] {$F_{6}$};
			\draw [blue, line width = 1 pt] (3.5,2.5) .. controls (3.4,1.8) and (3.8,1.2) .. (4.5,0.5) node [below] {$F_{5}-F_{6}$};
			\node (p7) at (1.67,2.87) [basept, label={[yshift=0pt] $q_{7}$}] {};
			\node (p8) at (1,2.5) [basept,label={[above] 		$q_{8}$}] {};
			\draw [line width = 0.8pt, ->]  (p8) edge (p7);
		\end{scope}

		\draw [->] (5.5,-5.5)--(7.5,-5.5) node[pos=0.5, below] {$\operatorname{Bl}_{q_{7}q_{8}}$};

		\begin{scope}[xshift = 0cm, yshift = -6.5cm]
			\draw [black, line width = 1pt ] (3.5,0) -- (-0.5,0) node [left] {$H_{y}$};
			\draw [blue, line width = 1pt ] (1.2,3.7) .. controls (0.8,3) and (0.3,2.5) .. (-0.5,2.5) node [left] {$H_{y} - F_{56}$};
			\draw [black, line width = 1pt ] (0,3) -- (0,-0.5) node [below] {$H_{x}$};
			\draw [red, line width = 1pt ] (4.5,1) .. controls (4,0.8) and (3,0.3) .. (3,-0.5) node [below] {$H_{x} - F_{5}$};

			\draw[blue, line width = 1pt] (3.5,1.9) circle [x radius=1.7cm, y radius = 1cm, rotate = 25] ;
			\node at (4,1.05) 	[circle, draw=white!100, fill=white!100, minimum size=1mm] {};

			\node at (4,3.4) [blue] {$2H_{x} + H_{y} - F_{123456} $};

			\draw [red, line width = 1 pt] (1.6,1.9) -- (2.5,2.1) node [pos = 0, left] {$F_{1}$};
			\draw [red, line width = 1 pt] (1.6,1.1) -- (2.5,1.6) node [pos = 0, left] {$F_{2}$};
			\draw [red, line width = 1 pt] (1.95,0.5) -- (2.65,1.2) node [pos = 0, left] {$F_{3}$};
			\draw [red, line width = 1 pt] (4.5,2.5) -- (5.4,2.3) node [right] {$F_{4}$};

			\draw [blue, line width = 1 pt] (0.4,3.4) -- (4.2,1.8) node [pos = 0, left] {$F_{6} - F_{7}$};
			\draw [blue, line width = 1 pt] (3.5,2.5) .. controls (3.4,1.8) and (3.8,1.2) .. (4.5,0.5) node [below] {$F_{5}-F_{6}$};
			\draw [blue, line width = 1 pt] (0.8,2.3)--(1.8,4) node [pos = 1, right] {$F_{7} - F_{8}$};
			\draw [red, line width = 1 pt] (1.1,4) -- (2,3.6) node [pos = 0, left] {$F_{8}$};


		\end{scope}

		\draw [->] (0.5,-2.2)--(0.5,-0.7) node[pos=0.5, right] {$\operatorname{Bl}_{q_{1}\cdots q_{8}}$};

	\end{tikzpicture}
	\caption{The Sakai Surface for the Hypergeometric Weight Recurrence}
	\label{fig:surface-hw}
\end{figure}

Given the base points $q_{1},\ldots,q_{8}$, the \emph{point configuration} is, essentially, the configuration of irreducible curves on which these points lie. 
These curves are in fact the irreducible components of some \emph{unique} (unless we have very special values of parameters that correspond
to an autonomous limit) bi-quadratic curve $\Gamma$. Let
the equation of $\Gamma$ in the $(X,Y)$-chart be
\begin{equation*}
	(a_{22} X^{2} + a_{12} X + a_{02}) Y^{2} + (a_{21} X^{2} + a_{11} X + a_{01}) Y + (a_{20} X^{2} + a_{10}X + a_{00}) = 0.
\end{equation*}
From the condition $q_{5}\in \Gamma$ we see that $a_{00} = 0$. To impose the condition that the infinitely close point (i.e., a point on an
exceptional divisor) $q_{6}\in \Gamma$, we rewrite this equation
in the $(u_{5},v_{5})$-chart (we should also include the $(U_{5},V_{5})$-chart, but unless it gives any new information, we omit those computations)
via the substitution $X = u_{5}$, $Y = u_{5} v_{5}$. The resulting equation factorizes,
\begin{equation*}
	u_{5}\Big( (a_{22} u_{5}^{2} + a_{12}u_{5} + a_{02}) u_{5}v_{5}^{2} + (a_{21} u_{5}^{2} + a_{11} u_{5} + a_{01})v_{5} + a_{20}u_{5} + a_{10} \Big) = 0.
\end{equation*}
This factorization corresponds to the decomposition of the \emph{total transform} of $\Gamma$ under
the blowup mapping $\operatorname{Bl}_{q_{5}}:\mathcal{X}_{q_{5}}\to \mathbb{P}^{1} \times \mathbb{P}^{1}$ into the irreducible components,
$\operatorname{Bl}_{q_{5}}^{-1}(\Gamma) = F_{5} + (\Gamma - F_{5})$, where $F_{5}$ is the central fiber of the blowup, and $\Gamma - F_{5}$
is the \emph{proper transform} of $\Gamma$. We then see that the condition $q_{6}\in \Gamma - F_{5}$ implies
$a_{10} = 0$. Continuing in this way through the degeneration cascade at $q_{5}$, as well as imposing the conditions $q_{i}\in \Gamma$ for $i=1,\ldots,4$,
we get the following equation for $\Gamma$:
\begin{equation*}
	\Gamma = V\left( Y s(X,Y)\right),\qquad \text{where}\quad s(X,Y) = X^{2} - \alpha \beta X^{2} Y + (n + \alpha + \beta) XY - Y.
\end{equation*}
The curve $\Gamma$ is in fact the \emph{pole divisor} of a \emph{symplectic form} $\omega$, and that is why $\Gamma$ is called an 
\emph{anti-canonical divisor}, $[\Gamma] = -\mathcal{K}_{\mathcal{X}}$.  Note also that the equation $s(X,Y)=0$ is nothing but the 
equation of the curve $\gamma$ in \eqref{eq:gamma-curve} written in the $(X,Y)$-chart. This is not surprising, since 
both $H_{y}-F_{5}-F_{6}$ and $\gamma$ are $-2$-curves that are irreducible components of the anti-canonical divisor, and so they 
are permuted by the mapping.

\begin{figure}[ht]
\begin{equation}\label{eq:d-roots-hw}
	\raisebox{-32.1pt}{\begin{tikzpicture}[
			elt/.style={circle,draw=black!100,thick, inner sep=0pt,minimum size=2mm}]
		\path 	(-1,1) 	node 	(d0) [elt, label={[xshift=-10pt, yshift = -10 pt] $\delta_{0}$} ] {}
		        (-1,-1) node 	(d1) [elt, label={[xshift=-10pt, yshift = -10 pt] $\delta_{1}$} ] {}
		        ( 0,0) 	node  	(d2) [elt, label={[xshift=10pt, yshift = -10 pt] $\delta_{2}$} ] {}
		        ( 1,1) 	node  	(d3) [elt, label={[xshift=10pt, yshift = -10 pt] $\delta_{3}$} ] {}
		        ( 1,-1) node 	(d4) [elt, label={[xshift=10pt, yshift = -10 pt] $\delta_{4}$} ] {};
		\draw [black,line width=1pt ] (d0) -- (d2) -- (d1)  (d3) -- (d2) -- (d4);
	\end{tikzpicture}} \qquad
			\begin{alignedat}{2}
			\delta_{0} &= \mathcal{F}_{5} - \mathcal{F}_{6}, &\qquad  \delta_{3} &= \mathcal{F}_{7} - \mathcal{F}_{8},\\
			\delta_{1} &= 2\mathcal{H}_{x} + \mathcal{H}_{y} - \mathcal{F}_{123456}, &\qquad  \delta_{4} &= \mathcal{H}_{y} - \mathcal{F}_{56}.\\
			\delta_{2} &= \mathcal{F}_{6} - \mathcal{F}_{7},
			\end{alignedat}
\end{equation}
	\caption{The Surface Root Basis for the Hypergeometric Weight Recurrence}
	\label{fig:d-roots-hw}
\end{figure}

This symplectic form $\omega$ in the affine $(X,Y)$-chart is given by
\begin{equation}\label{eq:sf-hw}
	\omega = k \frac{dX \wedge dY}{s(X,Y)\, Y} = k \frac{dX \wedge ds}{s(s - X^{2})},
\end{equation}
since
\begin{align*}
	ds &= (2X - 2 \alpha \beta XY + (n + \alpha + \beta)Y)\, dX - (\alpha \beta X^{2} - (n + \alpha + \beta)X + 1)\, dY,
	\intertext{and}
	\frac{X^{2} - s(X,Y)}{Y} &= \alpha \beta X^{2} - (n + \alpha + \beta)X + 1.
\end{align*}
This point configuration, the blowup diagram, and the decomposition of the anti-canonical divisor 
$-K_{\mathcal{X}} = \Gamma$ is shown on Figure~\ref{fig:surface-hw}.
Thus, we see that $-K_{\mathcal{X}}$ decomposes into irreducible components as follows:
\begin{equation}
	-K_{\mathcal{X}} = (2 H_{x} + H_{y} - F_{1} - F_{2} - F_{3} - F_{4} - F_{5} - F_{6}) + (H_{y} - F_{5} - F_{6}) + (F_{5} - F_{6})
	+ 2(F_{6} - F_{7}) + (F_{7} - F_{8}),
\end{equation}
whose intersection structure is given by the $D_{4}^{(1)}$ affine Dynkin diagram shown in Figure~\ref{fig:d-roots-hw}, where $\delta_{i} = [d_{i}]$.

Note that here the assignment of $d_{i}$ is arbitrary with the exception of $d_{2}$. Also, at this point we see that our equation is of
type d-$\dPain{D_{4}^{(1)}/D_{4}^{(1)}}$, and so our recurrence falls into the same family as the d-$\Pain{V}$ equation. However, to see whether
our recurrence is equivalent to d-$\Pain{V}$, we need to compare the dynamics. We describe the choice of the standard d-$\dPain{D_{4}^{(1)}/D_{4}^{(1)}}$
point configuration, choices of the root bases for the surface and the symmetry sub-lattices, and other data, in the Appendix; 
we follow \cite{KajNouYam:2017:GAOPE} in our conventions.

\subsection{Initial Geometry Identification} 
\label{sub:subsection_name}
To compare the application dynamics with the standard dynamics of d-$\Pain{V}$, we need to work with the same root bases. Thus, we begin by finding some change of 
basis of $\operatorname{Pic}(\mathcal{X})$ that will identify the surface roots between our recurrence and the standard example, and then use this change of basis 
to identify the symmetry roots and compare the translations. At this point, although we need to make some choices, we do not need to worry whether those choices
are correct, since they will be adjusted later on. 

\begin{lemma}\label{lem:change-basis-pre}
	The following change of basis of $\operatorname{Pic}(\mathcal{X})$ identifies the root bases 
	between the standard $D_{4}^{(1)}$ surface and the surface that we obtained for the hypergeometric weight recurrence:
	\begin{align*}
		\mathcal{H}_{x} & = \mathcal{H}_{g}, & \qquad 
		\mathcal{H}_{f} & = 2\mathcal{H}_{x} + \mathcal{H}_{y} - \mathcal{F}_{3} - \mathcal{F}_{4} - \mathcal{F}_{5} - \mathcal{F}_{6},\\
		\mathcal{H}_{y} & = \mathcal{H}_{f} + 2\mathcal{H}_{g}  - \mathcal{E}_{3} -\mathcal{E}_{4} - \mathcal{E}_{5} - \mathcal{E}_{6}, & 
		\qquad 	\mathcal{H}_{g} &= \mathcal{H}_{x},\\
		\mathcal{F}_{1} & = \mathcal{E}_{1}, & \qquad \mathcal{E}_{1} & = \mathcal{F}_{1},\\
		\mathcal{F}_{2} & = \mathcal{E}_{2}, & \qquad \mathcal{E}_{2} & = \mathcal{F}_{2},\\
		\mathcal{F}_{3} & = \mathcal{H}_{g} - \mathcal{E}_{6}, & \qquad \mathcal{E}_{3} & = \mathcal{H}_{x} - \mathcal{F}_{6},\\
		\mathcal{F}_{4} & = \mathcal{H}_{g} - \mathcal{E}_{5}, & \qquad \mathcal{E}_{4} & = \mathcal{H}_{x} - \mathcal{F}_{5},\\
		\mathcal{F}_{5} & = \mathcal{H}_{g} - \mathcal{E}_{4}, & \qquad \mathcal{E}_{5} & = \mathcal{H}_{x} - \mathcal{F}_{4},\\
		\mathcal{F}_{6} & = \mathcal{H}_{g} - \mathcal{E}_{3}, & \qquad \mathcal{E}_{6} & = \mathcal{H}_{x} - \mathcal{F}_{3},\\
		\mathcal{F}_{7} & = \mathcal{E}_{7}, & \qquad \mathcal{E}_{7} & = \mathcal{F}_{7},\\
		\mathcal{F}_{8} & = \mathcal{E}_{8}, & \qquad \mathcal{E}_{8} & = \mathcal{F}_{8}.
		\end{align*}
\end{lemma}

\begin{proof}
	Consider the surface sub-lattice root basis on Figure~\ref{fig:d-roots-hw} and compare it with the standard one on Figure~\ref{fig:d-roots-d41}. Since 
	the $D_{4}^{(1)}$ affine Dynkin diagram has the distinguished node $\delta_{2}$, we must have 
	\begin{equation*}
		\delta_{2} = \mathcal{F}_{6} - \mathcal{F}_{7} = \mathcal{H}_{g} - \mathcal{E}_{3} - \mathcal{E}_{7}.
	\end{equation*}	
	Thus, we can put $\mathcal{F}_{6} = \mathcal{H}_{g} - \mathcal{E}_{3}$ and $\mathcal{F}_{7} = \mathcal{E}_{7}$, and then, matching 
	$\mathcal{F}_{7} - \mathcal{F}_{8} = \mathcal{E}_{7} - \mathcal{E}_{8}$, we see that we can put $\mathcal{F}_{8} = \mathcal{E}_{8}$. Next, matching 
	$\mathcal{F}_{5} - \mathcal{F}_{6} = \mathcal{E}_{3} - \mathcal{E}_{4}$, we see that $\mathcal{F}_{5} = \mathcal{H}_{g} - \mathcal{E}_{4}$. Matching
	$\mathcal{H}_{y} - \mathcal{F}_{5} - \mathcal{F}_{6} =  \mathcal{H}_{f} - \mathcal{E}_{5} - \mathcal{E}_{6}$ we get 
	$\mathcal{H}_{y}$. The final node matching gives us the equation 
	$2 \mathcal{H}_{x}  - \mathcal{F}_{1234} = \mathcal{E}_{5} + \mathcal{E}_{6} - \mathcal{E}_{1} - \mathcal{E}_{2}$. Thus, we can put (again, at this point
	we do not worry about making the right choice)
	$\mathcal{F}_{1} = \mathcal{E}_{1}$, $\mathcal{F}_{2} = \mathcal{E}_{2}$, $\mathcal{E}_{5} = \mathcal{H}_{x} - \mathcal{F}_{4}$ and 
	$\mathcal{E}_{6} = \mathcal{H}_{x} - \mathcal{F}_{3}$, so that $\mathcal{H}_{x} = \mathcal{E}_{6} + \mathcal{F}_{3} = \mathcal{H}_{g}$. The inverse 
	change of basis is straightforward.	
\end{proof}


\subsection{The Symmetry Roots and the Translations} 
\label{sub:the_symmetry_roots_and_the_translations}
We are now in a position to compare the dynamics. Starting with the standard choice of the symmetry root basis \eqref{eq:a-roots-d4}			
and using the change of basis in Lemma~\ref{lem:change-basis-pre}, we get the symmetry roots for the applied problem shown on 
Figure~\ref{fig:a-roots-hw-pre}.	
\begin{figure}[ht]
\begin{equation}\label{eq:a-roots-hw-pre}			
	\raisebox{-32.1pt}{\begin{tikzpicture}[
			elt/.style={circle,draw=black!100,thick, inner sep=0pt,minimum size=2mm}]
		\path 	(-1,1) 	node 	(a0) [elt, label={[xshift=-10pt, yshift = -10 pt] $\alpha_{0}$} ] {}
		        (-1,-1) node 	(a1) [elt, label={[xshift=-10pt, yshift = -10 pt] $\alpha_{1}$} ] {}
		        ( 0,0) 	node  	(a2) [elt, label={[xshift=10pt, yshift = -10 pt] $\alpha_{2}$} ] {}
		        ( 1,1) 	node  	(a3) [elt, label={[xshift=10pt, yshift = -10 pt] $\alpha_{3}$} ] {}
		        ( 1,-1) node 	(a4) [elt, label={[xshift=10pt, yshift = -10 pt] $\alpha_{4}$} ] {};
		\draw [black,line width=1pt ] (a0) -- (a2) -- (a1)  (a3) -- (a2) -- (a4); 
	\end{tikzpicture}} \qquad
			\begin{alignedat}{2}
			\alpha_{0} &= \mathcal{H}_{y} - \mathcal{F}_{34}, &\qquad  \alpha_{3} &= 2 \mathcal{H}_{x} + \mathcal{H}_{y} - \mathcal{F}_{345678},\\
			\alpha_{1} &= \mathcal{F}_{1} - \mathcal{F}_{2}, &\qquad  \alpha_{4} &= \mathcal{F}_{3} - \mathcal{F}_{4}.\\
			\alpha_{2} &= \mathcal{F}_{4} - \mathcal{F}_{1},
			\end{alignedat}
\end{equation}
	\caption{The Symmetry Root Basis for the Hypergeometric Weight Recurrence (preliminary choice)}
	\label{fig:a-roots-hw-pre}	
\end{figure}

From the action of $\psi_{*}$ on $\operatorname{Pic}(\mathcal{X})$ given in Lemma~\ref{lem:dyn} we immediately see that the corresponding translation on the root lattice is 
\begin{equation}\label{eq:hw-transl-pre}
	\psi_{*}: \upalpha =  \langle \alpha_{0}, \alpha_{1}, \alpha_{2}, \alpha_{3}, \alpha_{4}  \rangle
	\mapsto \psi_{*}(\upalpha) = \upalpha + \langle 1,0,0,-1,0 \rangle \delta,
\end{equation}
which is \emph{different} from the standard translation vector $\langle 1,0,-1,1,0 \rangle$ given in \eqref{eq:dPv-rv-evol}. However, decomposing 
$\psi$ in terms of generators of the extended affine Weyl symmetry group, see Section~\ref{sub:the_extended_affine_weyl_symmetry_group}, 
and comparing it with the expression for $\varphi$ given in \eqref{eq:dPv-decomp},
\begin{equation}
	\psi = \sigma_{3}\sigma_{2} w_{3} w_{2} w_{4} w_{1} w_{2} w_{3},\qquad \varphi = \sigma_{3}\sigma_{2} w_{3} w_{0} w_{2} w_{4} w_{1} w_{2}, 
\end{equation}
we immediately see that $\psi = w_{3}\circ \varphi \circ w_{3}^{-1}$ (recall that $w_{3} \sigma_{3}\sigma_{2}  = \sigma_{3}\sigma_{2} w_{0}$ 
and that $w_{3}$ is an involution, $w_{3}^{-1} = w_{3}$). Thus, our dynamic is indeed equivalent to the standard d-$\Pain{V}$ equation, but the change of 
basis in Lemma~\ref{lem:change-basis-pre} needs to be adjusted by acting by $w_{3}$.

\begin{remark} At this point we verified the decompositions in \eqref{eq:hw-transl-pre} on the level of the Picard lattice. That is, if we use expressions
for symmetry roots in \eqref{eq:a-roots-hw-pre}	to define $w_{i}$ and $\sigma_{i}$ as acting on the Picard lattice, we get the expression for the mapping
$(\psi^{(n)})_{*}$ in Lemma~\ref{lem:dyn}. To obtain this decomposition on the level of actual maps first requires 
finding the change of variables that induces the change of basis in Lemma~\ref{lem:change-basis-pre} and then using it to
rewrite the birational representation in Section~\ref{sub:the_extended_affine_weyl_symmetry_group} in the application coordinates $(x,y)$. However, it is
too early to do so at this point, since the dynamic will not match; we do it later in Section~\ref{sub:the_change_of_coordinates}.
\end{remark}

\subsection{Final Geometry Identification} 
\label{sub:final_geometry_identification}

\begin{figure}[ht]
\begin{equation}\label{eq:a-roots-hw-fin}			
	\raisebox{-32.1pt}{\begin{tikzpicture}[
			elt/.style={circle,draw=black!100,thick, inner sep=0pt,minimum size=2mm}]
		\path 	(-1,1) 	node 	(a0) [elt, label={[xshift=-10pt, yshift = -10 pt] $\alpha_{0}$} ] {}
		        (-1,-1) node 	(a1) [elt, label={[xshift=-10pt, yshift = -10 pt] $\alpha_{1}$} ] {}
		        ( 0,0) 	node  	(a2) [elt, label={[xshift=10pt, yshift = -10 pt] $\alpha_{2}$} ] {}
		        ( 1,1) 	node  	(a3) [elt, label={[xshift=10pt, yshift = -10 pt] $\alpha_{3}$} ] {}
		        ( 1,-1) node 	(a4) [elt, label={[xshift=10pt, yshift = -10 pt] $\alpha_{4}$} ] {};
		\draw [black,line width=1pt ] (a0) -- (a2) -- (a1)  (a3) -- (a2) -- (a4); 
	\end{tikzpicture}} \qquad
			\begin{alignedat}{2}
			\alpha_{0} &= \mathcal{H}_{y} - \mathcal{F}_{34}, &\qquad  \alpha_{3} &= -2 \mathcal{H}_{x} - \mathcal{H}_{y} + \mathcal{F}_{345678},\\
			\alpha_{1} &= \mathcal{F}_{1} - \mathcal{F}_{2}, &\qquad  \alpha_{4} &= \mathcal{F}_{3} - \mathcal{F}_{4}.\\
			\alpha_{2} &= 2 \mathcal{H}_{x} + \mathcal{H}_{y}-\mathcal{F}_{135678},
			\end{alignedat}
\end{equation}
	\caption{The Symmetry Root Basis for the Hypergeometric Weight Recurrence (final choice)}
	\label{fig:a-roots-hw-fin}	
\end{figure}

\begin{lemma}\label{lem:change-basis-fin}
	After the change of basis of $\operatorname{Pic}(\mathcal{X})$ given by
	\begin{align*}
		\mathcal{H}_{x} & = \mathcal{H}_{f} + \mathcal{H}_{g} - \mathcal{E}_{7} - \mathcal{E}_{8}, & \quad 
		\mathcal{H}_{f} & = 2\mathcal{H}_{x} + \mathcal{H}_{y} - \mathcal{F}_{3} - \mathcal{F}_{4} - \mathcal{F}_{5} - \mathcal{F}_{6},\\
		\mathcal{H}_{y} & = 3\mathcal{H}_{f} + 2\mathcal{H}_{g}  - \mathcal{E}_{3} -\mathcal{E}_{4} - \mathcal{E}_{5} - \mathcal{E}_{6} - 
		2 \mathcal{E}_{7} - 2 \mathcal{E}_{8}, & 
		\qquad 	\mathcal{H}_{g} &= 3\mathcal{H}_{x} + \mathcal{H}_{y} - \mathcal{F}_{3} - \mathcal{F}_{4} - \mathcal{F}_{5} - \mathcal{F}_{6} - 
		\mathcal{F}_{7} - \mathcal{F}_{8},\\
		\mathcal{F}_{1} & = \mathcal{E}_{1}, & \qquad \mathcal{E}_{1} & = \mathcal{F}_{1},\\
		\mathcal{F}_{2} & = \mathcal{E}_{2}, & \qquad \mathcal{E}_{2} & = \mathcal{F}_{2},\\
		\mathcal{F}_{3} & = \mathcal{H}_{f} + \mathcal{H}_{g} - \mathcal{E}_{6} - \mathcal{E}_{7} - \mathcal{E}_{8}, & \quad 
		\mathcal{E}_{3} & = \mathcal{H}_{x} - \mathcal{F}_{6},\\
		\mathcal{F}_{4} & =  \mathcal{H}_{f} + \mathcal{H}_{g} - \mathcal{E}_{5} - \mathcal{E}_{7} - \mathcal{E}_{8}, & \quad 
		\mathcal{E}_{4} & = \mathcal{H}_{x} - \mathcal{F}_{5},\\
		\mathcal{F}_{5} & =  \mathcal{H}_{f} + \mathcal{H}_{g} - \mathcal{E}_{4} - \mathcal{E}_{7} - \mathcal{E}_{8}, & \quad 
		\mathcal{E}_{5} & = \mathcal{H}_{x} - \mathcal{F}_{4},\\
		\mathcal{F}_{6} & =  \mathcal{H}_{f} + \mathcal{H}_{g} - \mathcal{E}_{3} - \mathcal{E}_{7} - \mathcal{E}_{8}, & \quad 
		\mathcal{E}_{6} & = \mathcal{H}_{x} - \mathcal{F}_{3},\\
		\mathcal{F}_{7} & = \mathcal{H}_{f} - \mathcal{E}_{8}, & \quad 
		\mathcal{E}_{7} & = 2\mathcal{H}_{x} + \mathcal{H}_{y} - \mathcal{F}_{3} - \mathcal{F}_{4} - \mathcal{F}_{5} - \mathcal{F}_{6} - \mathcal{F}_{8},\\
		\mathcal{F}_{8} & = \mathcal{H}_{f} - \mathcal{E}_{7}, & \quad 
		\mathcal{E}_{8} & = 2\mathcal{H}_{x} + \mathcal{H}_{y} - \mathcal{F}_{3} - \mathcal{F}_{4} - \mathcal{F}_{5} - \mathcal{F}_{6} - \mathcal{F}_{7},
		\end{align*}
		the recurrence relations for variables $x_{n}$ and $y_{n}$ coincides with the standard d-$\Pain{V}$ discrete Painlev\'e equation given by
		\eqref{eq:dPv-std}. The resulting identification of the symmetry root bases (the surface root bases do not change) is shown in 
		Figure~\ref{fig:a-roots-hw-fin}.
\end{lemma}

Next we need to realize this change of basis on $\operatorname{Pic}(\mathcal{X})$ by an explicit change of coordinates. For that, it is convenient to first 
match the parameters between the applied problem and the reference example. This is done with the help of the \emph{Period Map}.

\subsection{The Period Map and the Identification of Parameters} 
\label{sub:the_period_map_and_the_identification_of_parameters}

The Period Map computation is similar to the slightly simpler standard case explained in Section~\ref{sub:the_period_map_and_the_root_variables}. 
Thus, we only state the result.

\begin{lemma}
	\qquad
	
	\begin{enumerate}[(i)]
		\item The residue of the symplectic form $\omega = k \frac{dX \wedge dY}{s(X,Y)\, Y} = k \frac{dX \wedge ds}{s(s - X^{2})}$ defined in \eqref{eq:sf-hw}
		along the irreducible components of the polar divisor is given by
		\begin{alignat*}{3}
			\operatorname{res}_{d_{0}} \omega &= - k \frac{dv_{5}}{v_{5}^{2}}, &\qquad
			\operatorname{res}_{d_{2}} \omega &= - k \frac{(n + \alpha + \beta)dv_{6}}{(v_{6} - 1)^{2}}, &\qquad
			\operatorname{res}_{d_{4}} \omega &= - k \frac{dX}{X^{2}}.\\
			\operatorname{res}_{d_{1}} \omega &=  k \frac{dX}{X^{2}}, &\qquad 
			\operatorname{res}_{d_{3}} \omega &= - k \frac{(c - 1)^{2} dv_{7}}{c}, &\qquad
		\end{alignat*}
		\item The root variables are given by 
		\begin{equation}\label{eq:root-vars-hw}
			a_{0} = k(\gamma - n - \alpha),\quad a_{1}= k(\alpha-1), \quad a_{2} = k(1 + n + \beta - \gamma),\quad a_{3} = -k(n + \beta),\quad a_{4} = k(\gamma - \beta).
		\end{equation}
		The normalization condition $a_{0} + a_{1} + 2a_{2} + a_{3} + a_{4} = 1$ then implies that $k=1$, and we get the following relations between our 
		parameters and the root variables:
		\begin{equation}\label{eq:par-match-app}
			\alpha = a_{1} + 1,\quad \beta = a_{0} + a_{1} + a_{2},\quad \gamma = 1 - a_{2} - a_{3},\quad n = a_{2} + a_{4} - 1.
		\end{equation}
		Note that the root variable evolution, which is the same as given in \eqref{eq:dPv-rv-evol}, is consistent with what we expect: 
		$\overline{\alpha} = \alpha$, $\overline{\beta} = \beta$, $\overline{\gamma} = \gamma$, and $\overline{n} = n+1$. Also observe that 
		we can not yet see the relationship between parameters $t$ and $c$ in this identification. 
		After we find the actual change of coordinates in the next section, we get that $ct=1$.
	\end{enumerate}
\end{lemma}


\subsection{The Change of Coordinates} 
\label{sub:the_change_of_coordinates}
We are now ready to prove Theorem~\ref{thm:coordinate-change}, which is the main result of the paper.
%
%
\begin{proof}(\textbf{Theorem 1})

	Since equations \eqref{eq:fg2xy} are simpler, we explain how to obtain them. Equations \eqref{eq:xy2fg} can then be either obtained in the same way 
	or by finding the explicit inverse change of variables from \eqref{eq:fg2xy}.
	
	From our change of basis, we see that 
	\begin{equation*}
		\mathcal{H}_{f}  = 2\mathcal{H}_{x} + \mathcal{H}_{y} - \mathcal{F}_{3} - \mathcal{F}_{4} - \mathcal{F}_{5} - \mathcal{F}_{6}.
	\end{equation*}
	Thus, $f(x,y)$ is a projective coordinate on a pencil of $(2,1)$-curves in the $(x,y)$-plane passing through points $q_{3}$, $q_{4}$, $q_{5}$, and $q_{6}$.
	Working in the $(X,Y)$-chart, we consider a generic $(2,1)$-curve $a_{00} + a_{01}X + a_{02}X^{2} + a_{10}Y + a_{11} XY + a_{12}X^{2}Y=0$. To pass 
	through $q_{5}(0,0)$ we much have $a_{00} = 0$, and to pass through the point $q_{6}(X=0,Y/X=0)$ we must have $a_{10} = 0$. Imposing conditions at $q_{3}$
	and $q_{4}$ gives us more constraints on the coefficients, and we get
	\begin{equation*}
		a_{11} \Big(Y (1 - X \beta)(1 - X \gamma)\Big) - a_{20} 
		\Big(Y (n + \alpha - \gamma) - X^{2} (\beta Y (\alpha \beta - \beta \gamma - n \gamma) - (\beta + \gamma))\Big)= 0.
	\end{equation*}
	The expressions at the coefficients $a_{11}$ and $a_{20}$ define two basis curves in the pencil and the coordinate $f(X,Y)$ is their ratio, up to a M\"obius
	transformation. When written in the $(x,y)$-chart, we get
	\begin{equation}
		f(x,y) = \frac{A  (x - \beta)(x - \gamma) + B \Big(x^{2} (n + \alpha - \gamma) + y (\beta + \gamma) - \beta (\alpha \beta - \beta \gamma - n \gamma)\Big)}{
		C (x - \beta)(x - \gamma) + D \Big(x^{2} (n + \alpha - \gamma) + y (\beta + \gamma) - \beta (\alpha \beta - \beta \gamma - n \gamma)\Big)},
	\end{equation}
	where the coefficients $A,B,C,D$ are still to be determined. To do that, we use the information about the exceptional divisor correspondence 
	in Lemma~\ref{lem:change-basis-fin}. For example, the condition $\mathcal{E}_{2} = \mathcal{F}_{2}$ means that 
	$(f,g)(q_{2}) = (f,g)(\alpha,-n \alpha) = (\infty,b_{2}) = p_{2}$, i.e., 
	\begin{equation*}
		F(\alpha, - n \alpha) = \frac{C + D(n + \alpha + \beta)}{A + B(n + \alpha + \beta)} = 0,\qquad\text{and so}\quad C = - D(n + \alpha + \beta).
	\end{equation*}
	The condition $\mathcal{E}_{6} = \mathcal{H}_{x} - \mathcal{F}_{3}$ means that $f(\beta,y) = B/D = 0$, and so $B=0$. As a result, after some simplifications,
	we get 
	\begin{equation*}
		f(x,y) = \frac{\tilde{A}(x-\beta)(x-\gamma)}{(x - \alpha) (x - \beta) - n x - y},
	\end{equation*}
	where $\tilde{A}$ is some proportionality constant. To find $\tilde{A}$, we use the condition $\mathcal{E}_{3} - \mathcal{E}_{4} = \mathcal{F}_{5} - \mathcal{F}_{6}$,
	which means that, after doing a sequence of substitutions to express $f$ in the $(u_{5},v_{5})$-chart and then restricting to $u_{5}=0$, the image of the 
	(proper transform $\mathcal{F}_{5} - \mathcal{F}_{6}$) of the exceptional divisor $\mathcal{F}_{5}$ should collapse to the point
	$p_{3}(t,\infty)$, i.e., $f(u_{5}=0,v_{5}) = t$. This results in $A=1$. Similarly, the condition $\mathcal{E}_{7} - \mathcal{E}_{8} = \mathcal{F}_{7} - \mathcal{F}_{8}$
	results in the relationship between $c$ and $t$, $ct=1$. Computing $g(x,y)$ is done exactly along the same lines, but the equations for the basis curves
	in the $\mathcal{H}_{g}$ pencil are more complicated, and so this computation is omitted. 
\end{proof}


\subsection{Partial Decompositions and Gauge Ambiguities} 
\label{sub:partial_decompositions}
In this section we want to make the following important point. Note that equation \eqref{eq:yn-evol} is a relation between $x_{n}$, $y_{n}$, and $y_{n+1}$ that 
we used to define the forward map $\psi_{1}:(x_{n},y_{n})\to (x_{n},y_{n+1})$. Similarly, equation \eqref{eq:xn-evol} is a relation between 
$x_{n-1}$, $x_{n}$, and $y_{n}$ that we used to define the backward map $\psi_{2}:(x_{n},y_{n})\to (x_{n-1},y_{n})$. In doing so we ignored possible 
$\mathbf{PGL}_{2}(\mathbb{C})\times \mathbf{PGL}_{2}(\mathbb{C})$ gauge group actions on both the domain and the range of the mappings. Thus, the mappings 
$\psi_{i}$ may not correspond exactly to elements of the birational representation of the symmetry group, where some normalization must be imposed to 
ensure the group structure on the level of the mappings. This point is essential, since we may not see the 
correct evolution of parameters in these partial maps. If necessary, this problem can be corrected using the action of the mappings on the Picard lattice
(that does not depend on the gauge actions) and the Period Map.

This issue can already be seen in the simpler model example of the difference Painlev\'e-V equation \eqref{eq:dPv-std}. 
This mapping can also be partially decomposed, in the natural way, as $\varphi= \varphi_{2}^{-1} \circ \varphi_{1}$, where
$\varphi_{1}$ is a forward mapping $\varphi_{1}:(f,g)\mapsto (\overline{f},-g)$ and $\varphi_{2}$ is a backward mapping 
$\varphi_{2}: (f,g)\mapsto (f,-\underline{g})$. Note that the additional negative sign (which is an example of the gauge group action mentioned above) 
is essential for the mappings $\varphi_{i}$ to be representable as a composition of elementary birational maps described in Theorems~\ref{thm:bir-weyl-d4}
and \ref{thm:bir-aut-d4}, where the normalization condition that we imposed in constructing the birational representation of $\widetilde{W}\left(D_{4}^{(1)}\right)$
is given by \eqref{eq:d4-param-norm}. In fact, there are two slightly different ways to write these mappings in terms of generators;
$\varphi = \varphi_{2}^{-1} \circ \varphi_{1} = \tilde{\varphi}_{2}^{-1} \circ \tilde{\varphi_{1}}$ (this is a direct calculation):
\begin{align}\label{eq:dPv-decomp1}
	\varphi_{1} &= \sigma_{3} \sigma_{2} w_{1} w_{2} w_{4} w_{1} w_{2}:(f,g)\mapsto (\overline{f},-g);\quad 
	\overline{a}_{0} = 1 - a_{0},\  \overline{a}_{1} = a_{1},\ \overline{a}_{2} = -a_{1} -a_{2},\ \overline{a}_{3} = 1-a_{3},\ 
	\overline{a}_{4} = -a_{4};\notag\\
	\varphi_{2} &=  w_{0} w_{3} w_{4}:(f,g)\mapsto (f,-\underline{g});\quad 
	\underline{a}_{0} =- a_{0},\  \underline{a}_{1} = a_{1},\ \underline{a}_{2} =1 -a_{1} -a_{2},\ \underline{a}_{3} =-a_{3},\ 
	\underline{a}_{4} = -a_{4},
\end{align}
or
\begin{align}\label{eq:dPv-decomp2}
	\tilde{\varphi}_{1} &= \sigma_{3} \sigma_{2} w_{1} w_{2} w_{4} w_{1} w_{2} w_{1}:(f,g)\mapsto (\overline{f},-g);\quad 
	\overline{a}_{0} = 1 - a_{0},\   \overline{a}_{1} = -a_{1},\ \overline{a}_{2} = -a_{2},\ \overline{a}_{3} = 1-a_{3},\ 
	\overline{a}_{4} = -a_{4};\notag\\
	\tilde{\varphi}_{2} &=  w_{0}w_{1}w_{3}w_{4}:(f,g)\mapsto (f,-\underline{g});\quad 
	\underline{a}_{0} =- a_{0},\  \underline{a}_{1} = -a_{1},\ \underline{a}_{2} =1  -a_{2},\ \underline{a}_{3} =-a_{3},\ 
	\underline{a}_{4} = -a_{4}.
\end{align}
Looking at the action of the mappings on the root variables $a_{i}$ it is clear that the individual mappings $\varphi_{1,2}$ do not correspond 
to translations on the symmetry sub-lattice; in fact, the need for the negative sign can be clearly seen at this point. The negative 
sign disappears, which is fairly typical, when we consider complete forward or backward steps in the dynamics, since those correspond to translations
\begin{equation*}
	(\underline{f},\underline{g})\xleftarrow{\varphi_{1}^{-1}} (f,-\underline{g}) \xleftarrow{\varphi_{2}}
	(f,g)
	\xrightarrow{\varphi_{1}}  (\overline{f},-g) 
	\xrightarrow{ \varphi_{2}^{-1} }(\overline{f},\overline{g}).
\end{equation*}

The same is true for the mappings $\psi_{i}$ from Lemma~\ref{lem:dyn}. 
Looking at the action of $(\psi_{1})_{*}$ on the symmetry roots \eqref{eq:a-roots-hw-fin}, we get
\begin{equation*}
(\psi_{1})_{*}(\alpha_{0}) = \delta - \alpha_{0},\ 
(\psi_{1})_{*}(\alpha_{1}) = - \alpha_{1},\ 
(\psi_{1})_{*}(\alpha_{2}) = \delta - \alpha_{2},\ 
(\psi_{1})_{*}(\alpha_{3}) = -\delta - \alpha_{3},\ 
(\psi_{1})_{*}(\alpha_{4}) = - \alpha_{4}. 
\end{equation*}			 
This immediately gives us the decomposition of $\psi_{1}$ and the action on the root variables:
\begin{equation}\label{eq:appl-map1}
	\psi_{1} = \sigma_{3} \sigma_{2} w_{0} w_{1} w_{2} w_{4} w_{1} w_{2} w_{3} w_{1};\quad 
	\overline{a}_{0} = 1 - a_{0},\  \overline{a}_{1} = -a_{1},\ \overline{a}_{2} = 1 -a_{2},\ \overline{a}_{3} = -1-a_{3},\ 
	\overline{a}_{4} = -a_{4}.
\end{equation}
Using \eqref{eq:par-match-app} we get the evolution of parameters (that is non-physical, since weight parameters should not change)
\begin{equation*}
	\overline{\alpha} = 2 - \alpha,\quad \overline{\beta} = 2- \beta,\quad \overline{\gamma} = 2 - \gamma,\quad \overline{n} = -(n+1),
\end{equation*}
which, in turn, gives us the forward evolution of the base points, which is \emph{different} from the evolution given in Lemma~\ref{lem:dyn}(a):
\begin{alignat*}{2}
\qquad &\overline{q}_{1}(1,(\alpha-1)(\beta-1)+(n+1)),&\qquad  &\overline{q}_{2}(2-\alpha,(n+1)(2-\alpha)),\hskip1.5in\phantom{a} \\
&\overline{q}_{3}(2-\beta,(n+1)(2-\beta)), &\qquad  &\overline{q}_{4}(2-\gamma,(\alpha - \gamma)(\beta - \gamma)+(n+1)(2-\gamma)),\\
&\overline{q}_{5}(\overline{X}=0,\overline{Y}=0)\leftarrow \mathrlap{\overline{q}_{6}\left(\overline{u}_{5} = \overline{X} = 0, \overline{v}_{5}
= \frac{\overline{Y}}{\overline{X}} = 0\right) \leftarrow
\overline{q}_{7}\left(\overline{u}_{6} = \overline{u}_{5} = 0, \overline{v}_{6} = \frac{\overline{v}_{5}}{\overline{u}_{5}}= \frac{c}{c-1}\right)} \\
&\qquad \leftarrow  \mathrlap{\overline{q}_{8}\left(\overline{u}_{7} = \overline{u}_{6} = 0,
\overline{v}_{7} = \frac{(c - 1)\overline{v}_{6} - c}{(c-1)\overline{u}_{6}} =
\frac{c\Big(c(3 - \alpha - \beta - n) - 3 - n + \gamma\Big)}{(c-1)^{2}}\right)}.
\end{alignat*}
Then the correct choice of the gauge to ensure that the mapping $\psi_{1}$ comes from the birational representation of the symmetry group is 
given by $\psi_{1}(x,y) = (2 - x, \overline{y} + 2(n+1))$, where $\overline{y}$ is given by \eqref{eq:fwd}.
This can either be deduced from the evolution of the base points \eqref{eq:appl-map1} or obtained directly from the birational representation of $\psi_{1}$ 
using the change of variables (\ref{eq:xy2fg}--\ref{eq:fg2xy}).

Similarly, for $(\psi_{2})_{*}$ the action on the symmetry roots is 
\begin{equation*}
(\psi_{1})_{*}(\alpha_{0}) = - \alpha_{0},\ 
(\psi_{1})_{*}(\alpha_{1}) = - \alpha_{1},\ 
(\psi_{1})_{*}(\alpha_{2}) = 2\delta - \alpha_{2},\ 
(\psi_{1})_{*}(\alpha_{3}) = -2\delta - \alpha_{3},\ 
(\psi_{1})_{*}(\alpha_{4}) = - \alpha_{4}. 
\end{equation*}			 
The resulting decomposition of $\psi_{2}$ and the action on the root variables is
\begin{equation}\label{eq:appl-map2}
	\psi_{2} = w_{4} w_{3} w_{2} w_{1} w_{0} w_{2} w_{4} w_{2} w_{1} w_{0} w_{2} w_{3} w_{1} w_{0};\quad 
	\underline{a}_{0} = - a_{0},\  \underline{a}_{1} = -a_{1},\ \underline{a}_{2} = 2 -a_{2},\ \underline{a}_{3} = -2-a_{3},\ 
	\underline{a}_{4} = -a_{4}.
\end{equation}
Using \eqref{eq:par-match-app} we get the evolution of parameters (that is again non-physical)
\begin{equation*}
	\underline{\alpha} = 2 - \alpha,\quad \underline{\beta} = 2- \beta,\quad \underline{\gamma} = 2 - \gamma,\quad \underline{n} = -n,
\end{equation*}
which, in turn, gives us the backward evolution of the base points, which is again \emph{different} from the evolution given in Lemma~\ref{lem:dyn}(b):
\begin{alignat*}{2}
\qquad &\underline{q}_{1}(1,(\alpha-1)(\beta-1)+n),&\qquad  &\underline{q}_{2}(2-\alpha,n(2-\alpha)),\hskip2.2in\phantom{a} \\
&\underline{q}_{3}(2-\beta,n(2-\beta)), &\qquad  &\underline{q}_{4}(2-\gamma,(\alpha - \gamma)(\beta - \gamma)+n(2-\gamma)),\\
&\underline{q}_{5}(\underline{X}=0,\underline{Y}=0)\leftarrow \mathrlap{\underline{q}_{6}\left(\underline{u}_{5} = \underline{X} = 0, \underline{v}_{5}
= \frac{\underline{Y}}{\underline{X}} = 0\right) \leftarrow
\underline{q}_{7}\left(\underline{u}_{6} = \underline{u}_{5} = 0, \underline{v}_{6} = \frac{\underline{v}_{5}}{\underline{u}_{5}}= \frac{c}{c-1}\right)} \\
&\qquad \leftarrow  \mathrlap{\underline{q}_{8}\left(\underline{u}_{7} = \underline{u}_{6} = 0,
\underline{v}_{7} = \frac{(c - 1)\underline{v}_{6} - c}{(c-1)\underline{u}_{6}} =
\frac{c\Big(c(4 - \alpha - \beta - n) - 2 - n + \gamma\Big)}{(c-1)^{2}}\right)}.
\end{alignat*}
Hence the correct choice of the gauge to ensure that the mapping $\psi_{2}$ comes from the birational representation of the symmetry group is 
given by $\psi_{2}(x,y) = (2 - \underline{x}, y + 2n)$, where $\underline{x}$ is given by \eqref{eq:back}.

Note that these gauge transformations cancel each other when we consider the full step. Indeed, let us define 
$\psi_{1}^{(n)}(x_{n},y_{n}) = (2 - x_{n}, y_{n+1} + 2(n+1))$ and $\psi_{2}^{(n)}(x_{n},y_{n}) = (2 - x_{n-1}, y_{n} + n)$. Then
\begin{equation*}
	\psi^{(n)}(x_{n},y_{n}) = \left(\psi_{2}^{(n+1)}\right)^{-1}\circ \psi_{1}^{(n)}(x_{n},y_{n}) = 
	\left(\psi_{2}^{(n+1)}\right)^{-1}(2 - x_{n}, y_{n+1} + 2(n+1)) = (x_{n+1},y_{n+1}).
\end{equation*}

\begin{remark} Given that both mappings $\psi$ and $\varphi$ decompose in a natural way, and that both mappings are equivalent, 
	it is reasonable to ask whether these decompositions are equivalent individually. This, unfortunately, is not the case. 
	Indeed, as can be seen from the above decompositions, $\psi_{1} = w_{3}\circ \varphi_{1}\circ w_{3}^{-1}$, but 
	$\psi_{2} = \psi_{1}\circ \varphi^{-1}$ and this can not really be simplified much further.	
\end{remark}



\section{Conclusions} 
\label{sec:conclusions}
In this paper we illustrated a systematic procedure on determining whether a second-order non-linear non-autonomous recurrence relation 
is a discrete Painlev\'e equation, and if so, how to reduce it to the standard form. We considered in detail an example from the theory of 
discrete orthogonal polynomials, where we showed that the evolution of recurrence coefficients for these polynomials is expressed in
terms of a particular solution of the standard difference Painlev\'e-V equation. However, it is clear that this approach can be easily
adapted to a wide range of other applied problems where discrete Painlev\'e equations appear.

\appendix

\section{Standard example of \lowercase{d}-$\dPain{D_{4}^{(1)}/D_{4}^{(1)}}$} 
\label{app:standard_example_of_dp-D4(1)}
In this section we review the standard example of discrete Painlev\'e equation of type 
d-$\dPain{D_{4}^{(1)}/D_{4}^{(1)}}$, also known as the d-$\Pain{V}$ equation. Note that this equation describes B\"acklund 
transformations of the usual differential $\Pain{VI}$ equation. We follow the standard reference \cite{KajNouYam:2017:GAOPE} for the 
choice of root bases and the form of the equation.


\subsection{The Point Configuration} 
\label{sub:the_point_configuration}
We start with the root basis of the surface sub-lattice that is given by the classes $\delta_{i} = [d_{i}]$ of the irreducible 
components of the anti-canonical divisor 
\begin{equation*}
	\delta = - \mathcal{K}_{\mathcal{X}} = 2 \mathcal{H}_{f} + 2 \mathcal{H}_{g} - \mathcal{E}_{1} 
	 - \mathcal{E}_{2} - \mathcal{E}_{3} - \mathcal{E}_{4} - \mathcal{E}_{5} - \mathcal{E}_{6} - \mathcal{E}_{7} - \mathcal{E}_{8}
	 = \delta_{0} + \delta_{1} + 2 \delta_{2} + \delta_{3} + \delta_{4}.
\end{equation*}
The intersection configuration of those roots is given by the Dynkin diagram of type $D_{4}^{(1)}$, as shown on Figure~\ref{fig:d-roots-d41}.
\begin{figure}[ht]
\begin{equation}\label{eq:d-roots-d41}			
	\raisebox{-32.1pt}{\begin{tikzpicture}[
			elt/.style={circle,draw=black!100,thick, inner sep=0pt,minimum size=2mm}]
		\path 	(-1,1) 	node 	(d0) [elt, label={[xshift=-10pt, yshift = -10 pt] $\delta_{0}$} ] {}
		        (-1,-1) node 	(d1) [elt, label={[xshift=-10pt, yshift = -10 pt] $\delta_{1}$} ] {}
		        ( 0,0) 	node  	(d2) [elt, label={[xshift=10pt, yshift = -10 pt] $\delta_{2}$} ] {}
		        ( 1,1) 	node  	(d3) [elt, label={[xshift=10pt, yshift = -10 pt] $\delta_{3}$} ] {}
		        ( 1,-1) node 	(d4) [elt, label={[xshift=10pt, yshift = -10 pt] $\delta_{4}$} ] {};
		\draw [black,line width=1pt ] (d0) -- (d2) -- (d1)  (d3) -- (d2) -- (d4); 
	\end{tikzpicture}} \qquad
			\begin{alignedat}{2}			
			\delta_{0} &= \mathcal{E}_{3} - \mathcal{E}_{4}, &\qquad  \delta_{3} &= \mathcal{E}_{7} - \mathcal{E}_{8},\\
			\delta_{1} &= \mathcal{H}_{f} - \mathcal{E}_{1} - \mathcal{E}_{2}, &\qquad  \delta_{4} &= \mathcal{H}_{f} - \mathcal{E}_{5} - \mathcal{E}_{6}.\\
			\delta_{2} &= \mathcal{H}_{g} - \mathcal{E}_{3} - \mathcal{E}_{7},\\[5pt]
			\delta & = \mathrlap{\delta_{0} + \delta_{1} + 2 \delta_{2} + \delta_{3} + \delta_{4}.} 
			\end{alignedat}
\end{equation}
	\caption{The Surface Root Basis for the standard d-$\dPain{D_{4}^{(1)}}$ point configuration}
	\label{fig:d-roots-d41}	
\end{figure}

Using the action of the $\mathbf{PGL}_{2}(\mathbb{C})\times \mathbf{PGL}_{2}(\mathbb{C})$ gauge group 
(i.e., the action of a M\"obius group on each of the factors of $\mathbb{P}^{1} \times \mathbb{P}^{1}$),
we can, without loss of generality, put $d_{i}$, with $\delta_{i} = [d_{i}]$ to be
\begin{equation*}
	d_{1} = V(F) = \{f = \infty\},\qquad d_{2} = V(G) = \{g = \infty\},\qquad d_{4} = V(f) = \{f = 0\},
\end{equation*} 
which then reduces the gauge group action to that of a three-parameter subgroup, $(f,g)\mapsto (\lambda f, \mu g + \nu)$. 
The corresponding point configuration and the Sakai surface are shown on Figure~\ref{fig:surface-d4}.
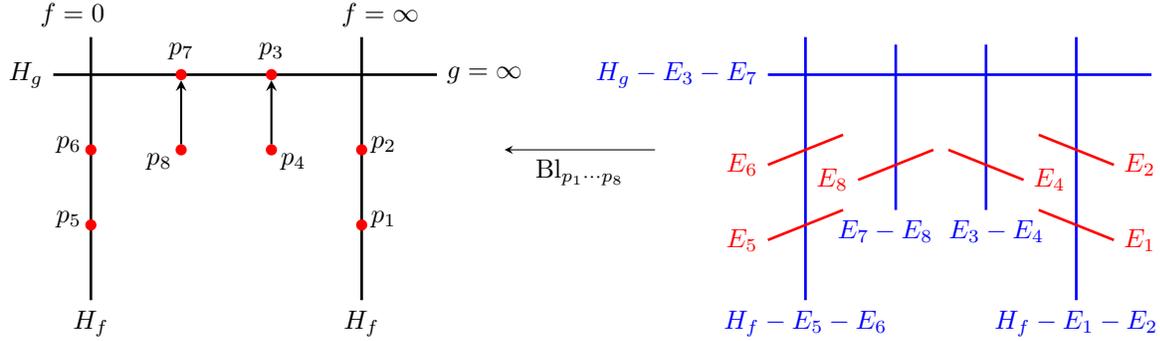
\begin{figure}[ht]
	\begin{tikzpicture}[>=stealth,basept/.style={circle, draw=red!100, fill=red!100, thick, inner sep=0pt,minimum size=1.2mm}]
		\begin{scope}[xshift = -1cm]
			\draw [black, line width = 1pt] 	(4.6,2.5) 	-- (-0.5,2.5)	node [left] {$H_{g}$} node[pos=0, right] {$g=\infty$};
			\draw [black, line width = 1pt] 	(0,3) -- (0,-0.5)			node [below] {$H_{f}$} node[pos=0, above, xshift=-7pt] {$f=0$};
			\draw [black, line width = 1pt] 	(3.6,3) -- (3.6,-0.5)		node [below] {$H_{f}$} node[pos=0, above, xshift=7pt] {$f=\infty$};
		
			\node (p1) at (3.6,0.5) [basept,label={[right] $p_{1}$}] {};
			\node (p2) at (3.6,1.5) [basept,label={[right] $p_{2}$}] {};
			\node (p3) at (2.4,2.5) [basept,label={[above] $p_{3}$}] {};
			\node (p4) at (2.4,1.5) [basept,label={[below right] $p_{4}$}] {};
			\node (p5) at (0,0.5) [basept,label={[left] $p_{5}$}] {};
			\node (p6) at (0,1.5) [basept,label={[left] $p_{6}$}] {};
			\node (p7) at (1.2,2.5) [basept,label={[above] $p_{7}$}] {};
			\node (p8) at (1.2,1.5) [basept,label={[below left] $p_{8}$}] {};
			\draw [line width = 0.8pt, ->] (p4) -- (p3);
			\draw [line width = 0.8pt, ->] (p8) -- (p7);
		\end{scope}
	
		\draw [->] (6.5,1.5)--(4.5,1.5) node[pos=0.5, below] {$\operatorname{Bl}_{p_{1}\cdots p_{8}}$};
	
		\begin{scope}[xshift = 8.5cm]
			\draw [blue, line width = 1pt] 	(4.6,2.5) 	-- (-0.5,2.5)	node [left] {$H_{g}-E_{3} - E_{7}$} node[pos=0, right] {};
			\draw [blue, line width = 1pt] 	(0,3) -- (0,-0.5)			node [below] {$H_{f}-E_{5}-E_{6}$} node[pos=0, above, xshift=-7pt] {};
			\draw [blue, line width = 1pt] 	(3.6,3) -- (3.6,-0.5)			node [below] {$H_{f}-E_{1}-E_{2}$} node[pos=0, above, xshift=7pt] {};
		
			\draw [red,line width = 1pt] (3.1,0.7) -- (4.1,0.3) node [right] {$E_{1}$};
			\draw [red,line width = 1pt] (3.1,1.7) -- (4.1,1.3) node [right] {$E_{2}$};
			\draw [red,line width = 1pt] (-0.5,0.3) -- (0.5,0.7) node [pos=0,left] {$E_{5}$};
			\draw [red,line width = 1pt] (-0.5,1.3) -- (0.5,1.7) node [pos=0,left] {$E_{6}$};
			\draw [blue,line width = 1pt] (2.4,0.7) -- (2.4,2.9) node [pos=0,below] {$\phantom{E}E_{3} - E_{4}$};
			\draw [blue,line width = 1pt] (1.2,0.7) -- (1.2,2.9) node [pos=0,below] {$E_{7} - E_{8}\phantom{E}$};
			\draw [red,line width = 1pt] (1.9,1.5) -- (2.9,1.1) node [right] {$E_{4}$};
			\draw [red,line width = 1pt] (0.7,1.1) -- (1.7,1.5) node [pos=0,left] {$E_{8}$};
		\end{scope}
	\end{tikzpicture}
	\caption{The model Sakai Surface for the d-$P\left(D_{4}^{(1)}/D_{4}^{(1)}\right)$ example}
	\label{fig:surface-d4}
\end{figure}

This point configuration can be parameterized by eight parameters $b_{1},\ldots, b_{8}$ as follows:
\begin{alignat*}{4}
	&p_{1}(\infty,b_{1}), 	&\quad &p_{2}(\infty,b_{2}), 	&\quad &p_{3}(b_{3},\infty)	&&\leftarrow p_{4}(b_{3},\infty;g(f-b_{3})=b_{4}),\\
	&p_{5}(0,b_{5}),		&\quad &p_{6}(0,b_{6}),			&\quad &p_{7}(b_{7},\infty)	&&\leftarrow p_{8}(b_{7},\infty;g(f-b_{7})=b_{8}).
\end{alignat*}
The three-parameter gauge group above acts on these configurations via
	\begin{equation}\label{eq:gauge-d4}
			\left(\begin{matrix}
				b_{1} & b_{2} & b_{3} & b_{4}\\
				b_{5} & b_{6} & b_{7} & b_{8}
			\end{matrix};  \begin{matrix}
				f \\ g
			\end{matrix}\right) \sim \left(\begin{matrix}
				\mu b_{1} + \nu & \mu b_{2} + \nu & \lambda b_{3} & \lambda \mu b_{4} \\
				\mu b_{5} + \nu & \mu b_{6} + \nu & \lambda b_{7} & \lambda \mu b_{8}
			\end{matrix}; \begin{matrix}
				\lambda f \\ \mu g +\nu
			\end{matrix}\right),\,\lambda,\mu\neq0,
	\end{equation}
and so the true number of parameters is five. The correct gauge-invariant parameterization is given by 
the \emph{root variables} that we now describe.

\subsection{The Period Map and the Root Variables} 
\label{sub:the_period_map_and_the_root_variables}
To define the root variables we begin by choosing a root basis in the 
\emph{symmetry sub-lattice} $Q = \Pi(R^{\perp}) \triangleleft \operatorname{Pic}(\mathcal{X})$ 
and defining the symplectic form $\omega$
whose polar divisor $-K_{\mathcal{X}}$ is the configuration of $-2$-curves shown on Figure~\ref{fig:surface-d4}. 
For the symmetry root basis we take the same basis as in \cite{KajNouYam:2017:GAOPE}, see Figure~\ref{fig:a-roots-d4}.

\begin{figure}[ht]
\begin{equation}\label{eq:a-roots-d4}			
	\raisebox{-32.1pt}{\begin{tikzpicture}[
			elt/.style={circle,draw=black!100,thick, inner sep=0pt,minimum size=2mm}]
		\path 	(-1,1) 	node 	(a0) [elt, label={[xshift=-10pt, yshift = -10 pt] $\alpha_{0}$} ] {}
		        (-1,-1) node 	(a1) [elt, label={[xshift=-10pt, yshift = -10 pt] $\alpha_{1}$} ] {}
		        ( 0,0) 	node  	(a2) [elt, label={[xshift=10pt, yshift = -10 pt] $\alpha_{2}$} ] {}
		        ( 1,1) 	node  	(a3) [elt, label={[xshift=10pt, yshift = -10 pt] $\alpha_{3}$} ] {}
		        ( 1,-1) node 	(a4) [elt, label={[xshift=10pt, yshift = -10 pt] $\alpha_{4}$} ] {};
		\draw [black,line width=1pt ] (a0) -- (a2) -- (a1)  (a3) -- (a2) -- (a4); 
	\end{tikzpicture}} \qquad
			\begin{alignedat}{2}
			\alpha_{0} &= \mathcal{H}_{f} - \mathcal{E}_{3} - \mathcal{E}_{4}, &\qquad  \alpha_{3} &= \mathcal{H}_{f} - \mathcal{E}_{7} - \mathcal{E}_{8},\\
			\alpha_{1} &= \mathcal{E}_{1} - \mathcal{E}_{2}, &\qquad  \alpha_{4} &= \mathcal{E}_{5} - \mathcal{E}_{6}.\\
			\alpha_{2} &= \mathcal{H}_{g} - \mathcal{E}_{1} - \mathcal{E}_{5},\\[5pt]
			\delta & = \mathrlap{\alpha_{0} + \alpha_{1} + 2 \alpha_{2} + \alpha_{3} + \alpha_{4}.} 
			\end{alignedat}
\end{equation}
	\caption{The Symmetry Root Basis for the standard d-$P\left(D_{4}^{(1)}\right)$ case}
	\label{fig:a-roots-d4}	
\end{figure}

A symplectic form $\omega\in -\mathcal{K}_{\mathcal{X}}$ such that $[\omega] = \delta_{0} + \delta_{1} + 2 \delta_{2} + \delta_{3} + \delta_{4}$ can be given in local 
coordinate charts as
\begin{equation}\label{eq:symp-form}
	\omega = k \frac{df\wedge dg}{f} = - k \frac{dF\wedge dg}{F} = - k 	\frac{df\wedge dG}{f G^{2}} = k \frac{dF \wedge dG}{F G^{2}} 
	= - k \frac{dU_{3}\wedge dV_{3}}{(b_{3} + U_{3} V_{3}) V_{3}} = - k \frac{dU_{7}\wedge d V_{7}}{(b_{7} + U_{7} V_{7}) V_{7}},
\end{equation}
where, as usual, $F = 1/f$, $G = 1/g$ are the coordinates centered at infinity, the blowup coordinates $(U_{i}, V_{i})$ at the points $p_{i}$, $i=3,7$,
are given by $f = b_{i} + U_{i}V_{i}$ and $G = V_{i}$, and $k$ is some non-zero proportionality constant that we normalize later. 
Then we have the following Lemma.

\begin{lemma}\label{lem-period_map-d4} 
	\qquad
	\begin{enumerate}[(i)]
		\item The residue of the symplectic form $\omega$ along the irreducible components of the polar divisor is given by
		\begin{equation}
			\operatorname{res}_{d_{0}} \omega = k \frac{d U_{3}}{b_{3}},\quad  
			\operatorname{res}_{d_{1}} \omega = - k dg,\quad  
			\operatorname{res}_{d_{2}} \omega = 0,\quad
			\operatorname{res}_{d_{3}} \omega = k \frac{d U_{7}}{b_{7}},\quad 
			\operatorname{res}_{d_{4}} \omega =  k dg.
		\end{equation}
		\item The root variables $a_{i}$ are given by 
		\begin{equation}\label{eq:d4-root_vars}
			a_{0} = - k \frac{b_{4}}{b_{3}},\quad a_{1} = k(b_{2} - b_{1}),\quad a_{2} = k (b_{1} - b_{5}),\quad a_{3} = - k \frac{b_{8}}{b_{7}}, \quad 
			a_{4} = k(b_{5} - b_{6}).
		\end{equation}
		It is convenient to take $k=-1$. We can then use the gauge action \eqref{eq:gauge-d4} to normalize $b_{5} = 0$, $b_{7} = 1$, and 
		$\chi(\delta) = a_{0} + a_{1} + 2a_{2} + a_{3} + a_{4} = 1$. In view of the relation of this example to $\Pain{VI}$, it is also convenient to denote 
		$b_{3}$ by $t$. Then we get the following parameterization of the point configuration in terms of root variables:
		\begin{equation}\label{eq:d4-root_var_par}
			b_{1} = - a_{2},\quad b_{2} = - a_{1} - a_{2},\quad b_{3} = t,\quad b_{4} = t a_{0},\quad b_{5} = 0,\quad b_{6} = a_{4},\quad b_{7} = 1,\quad b_{8} = a_{3}.
		\end{equation}  
		Note that if we use the notation $p_{34}\left(t(1 + \varepsilon a_{0}),1/\varepsilon\right)$, $p_{78}(1 + \varepsilon a_{3}, 1/\varepsilon)$ for the 
		degeneration cascades, we get exactly the parameterization of the point configuration in section 8.2.17 of \cite{KajNouYam:2017:GAOPE}.
	\end{enumerate}
\end{lemma}

\begin{proof}
	Part (a) is a standard computation in local charts. For example, with $d_{0} = E_{3} - E_{4} = V(V_{3})$ in the chart $(U_{3},V_{3})$, we get 
	\begin{equation*}
		\operatorname{res}_{d_{0}} \omega = \operatorname{res}_{V_{3} = 0} \left(- k \frac{dU_{3}\wedge dV_{3}}{(b_{3} + U_{3} V_{3}) V_{3}} \right) 
		= k \frac{d U_{3}}{b_{3}}.
	\end{equation*}
	Other computations in part (a) are similar.
	
	For part (b), first recall that the \emph{Period Map} $\chi: Q\to \mathbb{C}$ is defined on the simple roots $\alpha_{i}$, where 
	$a_{i}:= \chi(\alpha_{i})$ are called the \emph{root variables}, and then extended to the full symmetry sub-lattice
	by linearity. To compute the root variables $a_{i}$,  we proceed as follows, see \cite{Sak:2001:RSAWARSGPE} for details.
	\begin{itemize}
		\item First, we represent $\alpha_{i}$ as a difference of two effective divisors, 
		$\alpha_{i} = [C_{i}^{1}] - [C_{1}^{0}]$;
		\item second, note that there exists a \emph{unique} component $d_{k}$ of $-K_{\mathcal{X}}$ such that 
		$d_{k}\bullet C_{i}^{1} = d_{k}\bullet C_{i}^{0} = 1$, put $P_{i} =  d_{k}\cap C_{i}^{0}$ and 
		$Q_{i} =  d_{k}\cap C_{i}^{1}$:
		
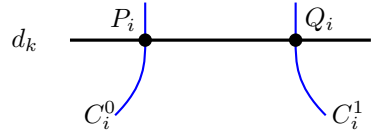
\begin{center}
			\begin{tikzpicture}[>=stealth, 
					elt/.style={circle,draw=black!100, fill=black!100, thick, inner sep=0pt,minimum size=1.5mm}]
					\draw[black, very thick] (0,0) -- (4,0);
					\draw[blue, thick] (1,0) -- (1,0.5);
					\draw[blue,thick] (1,0) .. controls (1,-0.3) and (1,-0.6) .. (0.6,-1);
					\draw[blue, thick] (3,0) -- (3,0.5);
					\draw[blue,thick] (3,0) .. controls (3,-0.3) and (3,-0.6) .. (3.4,-1);
					\node[style=elt] (P) at (1,0) {}; 		\node [above left] at (P) {$P_{i}$};
					\node[style=elt] (Q) at (3,0) {};		\node [above right] at (Q) {$Q_{i}$};
					\node at (-0.6,0) {$d_{k}$};
					\node at (0.4,-1) {$C_{i}^{0}$}; \node at (3.7,-1) {$C_{i}^{1}$};
					\end{tikzpicture}
		\end{center}
		\item then 
		\begin{equation*}
			\chi(\alpha_{i}) = \chi\left([C_{i}^{1}] - [C_{i}^{0}]\right) = 
			\int_{P_{i}}^{Q_{i}} \frac{ 1 }{ 2 \pi \mathfrak{i} }\oint_{d_{k}} \omega
			= \int_{P_{i}}^{Q_{i}} \operatorname{res}_{d_{k}} \omega,
		\end{equation*}
		where $\omega$ is the symplectic form defined by \eqref{eq:symp-form}.
	\end{itemize}
	We illustrate this procedure by computing the root variable $a_{0}$, the other computations are similar (see also \cite{DzhTak:2018:OSAOSGTODPE} for 
	more examples of such computations). First represent
	$\alpha_{0} = \mathcal{H}_{f} - \mathcal{E}_{3} - \mathcal{E}_{4} = [H_{f} - E_{3}] - [E_{4}]$. These two curves 
	intersect with the $d_{0}$ component of $\operatorname{div}(\omega)$, and so we get
	\begin{equation*}
			\raisebox{-1in}{\begin{tikzpicture}[>=stealth, 
					elt/.style={circle,draw=black!100, fill=black!100, thick, inner sep=0pt,minimum size=1.5mm}]

			\draw [blue, line width = 1pt] 	(0,2.5) 	-- (3.5,2.5) node[right] {$H_{g}-E_{3} - E_{7}$};
			\draw [blue,line width = 1pt] (1.4,-0.7) -- (1.4,3.2) node [above] {$d_{0} = E_{3} - E_{4}$};
			\draw [magenta, line width=1.6pt, ->] (1.4,2.5) -> (2,2.5) node [above] {$\qquad\qquad  u_{3} = f - b_{3}$};
			\draw [magenta, line width=1.6pt, ->] (1.4,2.5) -> (1.4,1.9) node [right] {$v_{3}$};
			\draw [red, line width = 1pt] 		(0,-1) .. controls (0,-0.5) and (0.5,0) .. (1,0) -- (3.5,0) node[right] {$C_{0}^{1} = H_{f}-E_{3}$};
			\draw [red,line width = 1pt] (0.9,1.3) -- (1.9,1.7) node [right] {$C_{0}^{0} = E_{4}$};
			\draw [magenta, line width=1.6pt, ->] (1.4,0) -> (2,0) node [above] {$\qquad\quad V_{3} = G$};
			\draw [magenta, line width=1.6pt, ->] (1.4,0) -> (1.4,0.6) node [right] {$U_{3}$};
			
			\node[style=elt] (P) at (1.4,1.5) {}; 	\node [below right] at (P) {$P_{0}(U_{3} = b_{4}, V_{3} = 0)$};
			\node[style=elt] (Q) at (1.4,0) {}; 	\node [below right] at (Q) {$Q_{0}(U_{3} = 0, V_{3} = 0)$};
			\end{tikzpicture}}\qquad 
			\begin{aligned}
				a_{0} &= \chi(\alpha_{0}) = \int_{P_{0}}^{Q_{0}} \operatorname{res}_{d_{0}} \omega  
				= k \int_{b_{4}}^{0} \frac{dU_{3}}{b_{3}} = - k \frac{b_{4}}{b_{3}}.
			\end{aligned}			
	\end{equation*}	
\end{proof}

\subsection{The Extended Affine Weyl Symmetry Group} 
\label{sub:the_extended_affine_weyl_symmetry_group}

For completeness, we also include here the description of the birational representation of the extended
affine Weyl symmetry group $\widetilde{W}\left(D_{4}^{(1)}\right) = \operatorname{Aut}\left(D_{4}^{(1)}\right) \ltimes W\left(D_{4}^{(1)}\right)$,
which is a \emph{semi-direct product} of the usual affine Weyl group $W\left(D_{4}^{(1)}\right)$ and the group of Dynkin diagram automorphisms
$\operatorname{Aut}\left(D_{4}^{(1)}\right)$.

The affine Weyl group $W\left(D_{4}^{(1)}\right)$ is defined in terms of generators $w_{i} = w_{\alpha_{i}}$ and relations that 
are encoded by the affine Dynkin diagram $D_{4}^{(1)}$,
\begin{equation*}
	W\left(D_{4}^{(1)}\right) = W\left(\raisebox{-20pt}{\begin{tikzpicture}[
			elt/.style={circle,draw=black!100,thick, inner sep=0pt,minimum size=1.5mm}]
		\path 	(-0.5,0.5) 	node 	(a0) [elt, label={[xshift=-10pt, yshift = -10 pt] $\alpha_{0}$} ] {}
		        (-0.5,-0.5) node 	(a1) [elt, label={[xshift=-10pt, yshift = -10 pt] $\alpha_{1}$} ] {}
		        ( 0,0) 	node  	(a2) [elt, label={[xshift=10pt, yshift = -10 pt] $\alpha_{2}$} ] {}
		        ( 0.5,0.5) 	node  	(a3) [elt, label={[xshift=10pt, yshift = -10 pt] $\alpha_{3}$} ] {}
		        ( 0.5,-0.5) node 	(a4) [elt, label={[xshift=10pt, yshift = -10 pt] $\alpha_{4}$} ] {};
		\draw [black,line width=1pt ] (a0) -- (a2) -- (a1)  (a3) -- (a2) -- (a4); 
	\end{tikzpicture}} \right)
	=
	\left\langle w_{0},\dots, w_{4}\ \left|\ 
	\begin{alignedat}{2}
    w_{i}^{2} = e,\quad  w_{i}\circ w_{j} &= w_{j}\circ w_{i}& &\text{ when 
   				\raisebox{-0.08in}{\begin{tikzpicture}[
   							elt/.style={circle,draw=black!100,thick, inner sep=0pt,minimum size=1.5mm}]
   						\path   ( 0,0) 	node  	(ai) [elt] {}
   						        ( 0.5,0) 	node  	(aj) [elt] {};
   						\draw [black] (ai)  (aj);
   							\node at ($(ai.south) + (0,-0.2)$) 	{$\alpha_{i}$};
   							\node at ($(aj.south) + (0,-0.2)$)  {$\alpha_{j}$};
   							\end{tikzpicture}}}\\
    w_{i}\circ w_{j}\circ w_{i} &= w_{j}\circ w_{i}\circ w_{j}& &\text{ when 
   				\raisebox{-0.17in}{\begin{tikzpicture}[
   							elt/.style={circle,draw=black!100,thick, inner sep=0pt,minimum size=1.5mm}]
   						\path   ( 0,0) 	node  	(ai) [elt] {}
   						        ( 0.5,0) 	node  	(aj) [elt] {};
   						\draw [black] (ai) -- (aj);
   							\node at ($(ai.south) + (0,-0.2)$) 	{$\alpha_{i}$};
   							\node at ($(aj.south) + (0,-0.2)$)  {$\alpha_{j}$};
   							\end{tikzpicture}}}
	\end{alignedat}\right.\right\rangle. 
\end{equation*} 
The natural action of this group on $\operatorname{Pic}(\mathcal{X})$ is given by reflections in the 
roots $\alpha_{i}$, 
\begin{equation}\label{eq:root-refl}
	w_{i}(\mathcal{C}) = w_{\alpha_{i}}(\mathcal{C}) = \mathcal{C} - 2 
	\frac{\mathcal{C}\bullet \alpha_{i}}{\alpha_{i}\bullet \alpha_{i}}\alpha_{i}
	= \mathcal{C} + \left(\mathcal{C}\bullet \alpha_{i}\right) \alpha_{i},\qquad \mathcal{C}\in \operatorname{Pic(\mathcal{X})},
\end{equation}
which can be extended to an action on point configurations by elementary birational maps (which lifts to 
isomorphisms $w_{i}: \mathcal{X}_{\mathbf{b}}\to \mathcal{X}_{\overline{\mathbf{b}}}$ on the family of Sakai's surfaces),
this is known as a birational representation of $W\left(D_{4}^{(1)}\right)$.

\begin{remark}
\label{rem:period-action}
Recall that for an arbitrary $w\in \widetilde{W}\left(D_{4}^{(1)}\right)$, the action of $w$ on the root variables is \emph{inverse} to its action
on the roots. This is not essential for the generating reflections, that are involutions, but it is important for composed maps.
\end{remark}

\begin{theorem}\label{thm:bir-weyl-d4}
	Reflections $w_{i}$ on $\operatorname{Pic}(\mathcal{X})$ are induced by the elementary 
	birational mappings given below and also denoted by $w_{i}$, on the family $\mathcal{X}_{\mathbf{b}}$. To ensure the group structure, 
	we require that each map preserves our normalization
	\begin{equation}\label{eq:d4-param-norm}
		\left(\begin{matrix}
 			{b}_{1} & {b}_{2} & {b}_{3} & {b}_{4}\\
 			{b}_{5} & {b}_{6} & {b}_{7} & {b}_{8}
 		\end{matrix}\right)=
		\left(\begin{matrix}
 			{b}_{1} & {b}_{2} 	& t 	& {b}_{4}\\
 			0 		& {b}_{6} 	& 1		& {b}_{8}
 		\end{matrix}\right)=
		\left(\begin{matrix}
 			-a_{2} & -a_{1} - a_{2} & t & t a_{0}\\
 			0 & a_{4} & 1 & a_{3}
 		\end{matrix}\right).
	\end{equation}
	We give the action of the mappings both on parameters $b_{i}$ related to the parameterization of point configurations, and on the root 
	variables (note that the parameter $t$ can also change when we consider the Dynkin diagram automorphisms, so it is convenient to include it 
	among the root variables). 	For the initial configuration 
	\begin{equation*}
		\left(\begin{matrix}
 			{b}_{1} & {b}_{2} 	& t 	& {b}_{4}\\
 			0 		& {b}_{6} 	& 1		& {b}_{8}
 		\end{matrix}
		\ ;\   
		\begin{matrix}
 			f \\ g
		\end{matrix}\right) = \left(\begin{matrix}
 			a_{0} & a_{1} 	& a_{2}\\
 			 a_{3} & a_{4} & t
 		\end{matrix}
		\ ;\ 
		\begin{matrix}
 			f \\ g
		\end{matrix}\right),
	\end{equation*}
	the action of $w_{i}$ is given by the following expressions:
	\begin{alignat}{2}
		w_{0}&: 
		\left(\begin{matrix}
 			{b}_{1} - \frac{b_{4}}{t} & {b}_{2} - \frac{b_{4}}{t} 	& t 	& -{b}_{4}\\
 			0 		& {b}_{6} 	& 1		& {b}_{8}
		\end{matrix}\ ;\  \begin{matrix}
		 			f\\ g - \frac{b_{4} f}{t(f-t)}
		 		\end{matrix}\right) &&= \left(\begin{matrix}
 			-a_{0} & a_{1} 	& a_{0} + a_{2}\\
 			 a_{3} & a_{4} & t
 		\end{matrix}\ ;\  
		\begin{matrix}
 			f \\ g - \frac{a_{0}f}{f-t}
		\end{matrix}\right),\\
		w_{1}&: 
		\left(\begin{matrix}
 			{b}_{2} & {b}_{1} & t 	& {b}_{4}\\
 			0 		& {b}_{6} 	& 1		& {b}_{8}
		\end{matrix}\ ;\  \begin{matrix}
		 			f\\ g 
		 		\end{matrix}\right) &&= \left(\begin{matrix}
 			a_{0} & -a_{1} 	& a_{1} + a_{2}\\
 			 a_{3} & a_{4} & t
 		\end{matrix}\ ;\  
		\begin{matrix}
 			f \\ g 
		\end{matrix}\right),\\
		w_{2}&: 
		\left(\begin{matrix}
 			-{b}_{1}  & {b}_{2} - b_{1} 	& t 	& {b}_{4} - t b_{1}\\
 			0 		& {b}_{6} - b_{1} 	& 1		& {b}_{8} - b_{1}
		\end{matrix}\ ;\  \begin{matrix}
		 			f - \frac{b_{1} f}{g}\\ g - b_{1}
		 		\end{matrix}\right) &&= \left(\begin{matrix}
 		a_{0} + a_{2} & a_{1}+ a_{2} 	& - a_{2}\\
 			 a_{2} + a_{3} & a_{2} + a_{4} & t
 		\end{matrix}\ ;\  
		\begin{matrix}
 			f  + \frac{a_{2} f}{g}\\ g + a_{2}
		\end{matrix}\right),\\
		w_{3}&: 
		\left(\begin{matrix}
 			{b}_{1} - b_{8} & {b}_{2} - b_{8} 	& t 	& {b}_{4}\\
 			0 		& {b}_{6} 	& 1		& -{b}_{8}
		\end{matrix}\ ;\  \begin{matrix}
		 			f\\ g - \frac{b_{8} f}{f-1}
		 		\end{matrix}\right) &&= \left(\begin{matrix}
 			a_{0} & a_{1} 	& a_{2} + a_{3}\\
 			 -a_{3} & a_{4} & t
 		\end{matrix}\ ;\  
		\begin{matrix}
 			f \\ g - \frac{a_{3}f}{f-1}
		\end{matrix}\right),\\
		w_{4}&: 
		\left(\begin{matrix}
 			{b}_{1} - b_{6} & {b}_{2} - b_{6} 	& t 	& b_{4}\\
 			0 		& -{b}_{6} 	& 1		& {b}_{8}
		\end{matrix}\ ;\  \begin{matrix}
		 			f\\ g - b_{6}
		 		\end{matrix}\right) &&= \left(\begin{matrix}
 			a_{0} & a_{1} 	& a_{2} + a_{4}\\
 			 a_{3} & -a_{4} & t
 		\end{matrix}\ ;\  
		\begin{matrix}
 			f \\ g - a_{4}
		\end{matrix}\right).		
	\end{alignat}	
\end{theorem}

\begin{proof}
The proof is standard, see \cite{DzhTak:2018:OSAOSGTODPE} for careful explanations, but to make this paper self-contained, we 
briefly outline one such computation. 
The reflection $w_{0}$ in the root $\alpha_{0} = \mathcal{H}_{f} - \mathcal{E}_{3} - \mathcal{E}_{4}$ acts on 
$\operatorname{Pic}(\mathcal{X})$ by
\begin{equation*}
	w_{0}(\mathcal{H}_{f}) = \mathcal{H}_{f},\ 
	w_{0}(\mathcal{H}_{g}) = \mathcal{H}_{f} + \mathcal{H}_{g} - \mathcal{E}_{3} - \mathcal{E}_{4},\ 
	w_{0}(\mathcal{E}_{3}) = \mathcal{H}_{f} - \mathcal{E}_{4},\  w_{0}(\mathcal{E}_{4}) = \mathcal{H}_{f} - \mathcal{E}_{3},
	\  w_{0}(\mathcal{E}_{i}) = \mathcal{E}_{i}, i\neq 3,4.
\end{equation*}
Thus, we are looking for a mapping $w_{0}: \mathcal{X}_{\mathbf{b}}\to \mathcal{X}_{\overline{\mathbf{b}}}$ that is given in the affine 
chart $(f,g)$ by a formula $w_{0}(f,g) = (\overline{f},\overline{g})$ so that 
\begin{equation*}
	w_{0}^{*}(\mathcal{H}_{\overline{f}}) = \mathcal{H}_{f},\qquad w_{0}^{*}(\mathcal{H}_{\overline{g}}) = 
	\mathcal{H}_{f} + \mathcal{H}_{g} - \mathcal{E}_{3} - \mathcal{E}_{4}.
\end{equation*}
Thus, up to M\"obius transformations, $\overline{f}$ coincides with $f$ and $\overline{g}$ is a coordinate on a pencil of $(1,1)$-curves
passing through the degeneration cascade $p_{3}(b_{3},\infty)\leftarrow p_{4}(b_{3},\infty;g(f-b_{3})=b_{4})$. 
Let $|H_{\overline{g}}| = \{A fg + Bf + C g + D = 0\}$. Then $p_{3}$ imposes the condition $A b_{3} + C = 0$, and so 
$|H_{\overline{g}}| = \{A (f-b_{3})g + Bf  + D = 0\}$. Point $p_{4}$ then imposes the condition $A b_{4} + B b_{3} + D = 0$
and we see that the basis of the pencil $|H_{\overline{g}}|$ is given by $ (f-b_{3})g - b_{4}$ and $f - b_{3}$.
Taking the M\"obius transformations into account, we get
\begin{equation*}
	\overline{f} = \frac{A f + B}{C f + D},\quad \overline{g} = \frac{K((f - b_{3})g - b_{4}) + L (f - b_{3})}{M((f - b_{3})g - b_{4}) + N (f - b_{3})},
\end{equation*}
where $A,\ldots,N$ are some constants to be determined. We also know that the root variables change as $\overline{a}_{0} = - a_{0}$, 
$\overline{a}_{2} = a_{0} + a_{2}$, and $\overline{a}_{i} = a_{i}$ otherwise. This then gives us the evolution of parameters $b_{i}$, e.g., 
$\overline{b}_{1} = - \overline{a}_{2} = - a_{2} - a_{0} = b_{1} - b_{4}/t$ (recall that $t = b_{3}$), and so on. The constants $A,\ldots,N$ can be determined from the action 
of $w_{0}$ on exceptional divisors. For example, $w_{0}(\overline{\mathcal{E}}_{5}) = \mathcal{E}_{5}$ is equivalent to 
\begin{equation*}
	(\overline{f},\overline{g})(0,0) = (0,0)\quad\implies\quad B = 0,\quad L = - K \frac{b_{4}}{b_{3}},
\end{equation*}
$w_{0}(\overline{\mathcal{E}}_{7}) = \mathcal{E}_{7}$ implies that $M=0$,  then $w_{0}(\overline{\mathcal{E}}_{1}) = \mathcal{E}_{1}$ gives 
\begin{equation*}
	(\overline{f},\overline{g})(\infty,b_{1}) = (\infty,\overline{b}_{1}) = (\infty, b_{1} - b_{4}/t)\quad\implies\quad C = 0,\quad K/N = 1,
\end{equation*}
and so on.
 
\end{proof}

Let us now describe the group of Dynkin diagram automorphisms. It is clear that $\operatorname{Aut}\left(D_{4}^{(1)}\right)\simeq \mathcal{S}_{4}$, so we only describe
three transpositions that generate the whole group.

\begin{theorem}\label{thm:bir-aut-d4}
	Consider the following generators $\sigma_{1},\ldots ,\sigma_{3}$ of $\operatorname{Aut}\left(D_{4}^{(1)}\right)$ 
	that act on the symmetry and the surface root bases as follows (here we use the standard cycle 
	notations for permutations):
	\begin{equation}
		\sigma_{1} = (\alpha_{3}\alpha_{4})=(\delta_{3}\delta_{4}),\qquad \sigma_{2} = (\alpha_{0}\alpha_{3})=(\delta_{0}\delta_{3}), \qquad
		\sigma_{3} = (\alpha_{1}\alpha_{4})=(\delta_{1}\delta_{4}).
	\end{equation}
	Then $\sigma_{i}$ act on the Picard lattice as
	\begin{equation*}
			\sigma_{1} = (\mathcal{E}_{6}\mathcal{E}_{8})w_{\rho},\qquad \sigma_{2} = (\mathcal{E}_{3}\mathcal{E}_{7})(\mathcal{E}_{4}\mathcal{E}_{8}), \qquad
			\sigma_{3} = (\mathcal{E}_{1}\mathcal{E}_{5})(\mathcal{E}_{2}\mathcal{E}_{6}),			
	\end{equation*}
	where $w_{\rho}$ is a reflection  \eqref{eq:root-refl} in the root $\rho = \mathcal{H}_{f} - \mathcal{E}_{5} - \mathcal{E}_{7}$ (note also that a transposition 
	$(\mathcal{E}_{i} \mathcal{E}_{j})$ is induced by a reflection in the root $\mathcal{E}_{i} - \mathcal{E}_{j}$).
	The induced elementary 	birational mappings are then given by the following expressions:
	\begin{alignat}{2}
		\sigma_{1}&: 
		\left(\begin{matrix}
 			{b}_{1} & {b}_{2} & 1-t 	& \frac{(1-t)b_{4}}{t}\\
 			0 		& {b}_{8} 	& 1		& {b}_{6}
		\end{matrix}\ ;\  \begin{matrix}
		 			1 - f\\ \frac{(f-1)g}{f}
		 		\end{matrix}\right) &&= \left(\begin{matrix}
 			a_{0} & a_{1} & a_{2}\\
 			a_{4} & a_{3} & 1 - t 
 		\end{matrix}\ ;\  
		\begin{matrix}
 			1 - f\\ \frac{(f-1)g}{f}
		\end{matrix}\right),\label{eq:sigma1}\\
		\sigma_{2}&: 
		\left(\begin{matrix}
 			{b}_{1} & {b}_{2} & \frac{1}{t} & \frac{b_{8}}{t}\\
 			0 		& {b}_{6} 	& 1		& \frac{b_{4}}{t}
		\end{matrix}\ ;\  \begin{matrix}
		 			\frac{f}{t}\\ g 
		 		\end{matrix}\right) &&= \left(\begin{matrix}
 			a_{3} & a_{1} 	& a_{2}\\
 			a_{0} & a_{4} & \frac{1}{t}
 		\end{matrix}\ ;\  
		\begin{matrix}
 			\frac{f}{t} \\ g 
		\end{matrix}\right),\\
		\sigma_{3}&: 
		\left(\begin{matrix}
 			{b}_{1}  & {b}_{1} - b_{6} 	& \frac{1}{t} 	& \frac{b_{4}}{t^{2}}\\
 			0 		& {b}_{1} - b_{2} 	& 1		& {b}_{8}
		\end{matrix}\ ;\  \begin{matrix}
		 			\frac{1}{f}\\ b_{1} - g
		 		\end{matrix}\right) &&= \left(\begin{matrix}
 		a_{0} & a_{4}	&  a_{2}\\
 		a_{3} & a_{1} & \frac{1}{t}
 		\end{matrix}\ ;\  
		\begin{matrix}
 			\frac{1}{f}\\ - g - a_{2}
		\end{matrix}\right).		
	\end{alignat}	
\end{theorem}

\begin{proof}
We briefly outline the proof for $\sigma_{1}$. First, note that we define the action of $\sigma_{1}$ on the symmetry roots $\alpha_{i}$ and then 
try to deduce its action on both the surface roots $\delta_{i}$ and also on all of $\operatorname{Pic}(\mathcal{X})$. Moreover, $\sigma_{1}$ is an 
involution and from $\sigma_{1}: \alpha_{3} \leftrightarrow \alpha_{4}$  we see that it is natural to ask that 
$\sigma_{1}:\mathcal{H}_{f}-\mathcal{E}_{7}\leftrightarrow \mathcal{E}_{5}$ and $\sigma_{1}: \mathcal{E}_{6}\leftrightarrow \mathcal{E}_{8}$.
Looking at the surface roots $\delta_{i}$ we see that then $\sigma_{1}$ permutes the roots $\delta_{3}$ and $\delta_{4}$, i.e., 
$\sigma_{1}:\mathcal{H}_{f}-\mathcal{E}_{5}\leftrightarrow \mathcal{E}_{7}$. Thus, $\sigma_{1}$ fixes $\mathcal{H}_{f}$. Requiring 
that $\sigma_{1}$ fixes $\alpha_{i}$ and $\delta_{i}$ for $i=0,1,2$ implies that $\sigma_{1}$ fixes $\mathcal{E}_{i}$ for $i=1,\ldots,4$.
From this it immediately follows that $\sigma_{1}(\mathcal{H}_{g}) = \mathcal{H}_{f} + \mathcal{H}_{g} - \mathcal{E}_{5} - \mathcal{E}_{7}$.
It is now easy to see that $\sigma_{1}$ can be realized as a composition of two reflections in the roots $\mathcal{E}_{6} - \mathcal{E}_{8}$ and
$\mathcal{H}_{f} - \mathcal{E}_{5} - \mathcal{E}_{7}$, 
\begin{equation*}
	\sigma_{1} = w_{\mathcal{E}_{6}-\mathcal{E}_{8}} w_{\mathcal{H}_{f}- \mathcal{E}_{5}-\mathcal{E}_{7}}: 
	\mathcal{H}_{g}\mapsto \mathcal{H}_{f} + \mathcal{H}_{g} - \mathcal{E}_{5} - \mathcal{E}_{7},\ \mathcal{E}_{5}\mapsto \mathcal{H}_{f} - \mathcal{E}_{7},\ 
	\mathcal{E}_{6}\mapsto \mathcal{E}_{8},\ \mathcal{E}_{7}\mapsto \mathcal{H}_{f} - \mathcal{E}_{5},\ 
	\mathcal{E}_{8}\mapsto \mathcal{E}_{6},
\end{equation*}
and the remaining generators of $\operatorname{Pic}(\mathcal{X})$ are fixed. The rest of the proof is now similar to the previous Theorem. Let
$\sigma_{1}: \mathcal{X}_{\mathbf{b}}\to \mathcal{X}_{\overline{\mathbf{b}}}$ be written in the affine 
chart $(f,g)$ as $\sigma_{1}(f,g) = (\overline{f},\overline{g})$. Requiring that
\begin{equation*}
	\sigma_{1}^{*}(\mathcal{H}_{\overline{f}}) = \mathcal{H}_{f},\qquad \sigma_{1}^{*}(\mathcal{H}_{\overline{g}}) = 
	\mathcal{H}_{f} + \mathcal{H}_{g} - \mathcal{E}_{5} - \mathcal{E}_{7}
\end{equation*}
we get the mapping up to M\"obius transformation,
\begin{equation*}
	\overline{f} = \frac{A f + B}{C f + D},\quad \overline{g} = \frac{K(g(f-b_{7}) + b_{5} b_{7}) + L f}{M(g(f-b_{7}) + b_{5} b_{7}) + N f} = 
	\frac{K g (f-1) + L f}{M g(f - 1) + N f},
\end{equation*}
where we used the normalization $b_{5} =0$ and $b_{7} = 1$; as usual, $A,\ldots,N$ are some constants to be determined. 
We also know that $\overline{a}_{3} = a_{4}$, $\overline{a}_{4} = a_{3}$, 
and $\overline{a}_{i} = a_{i}$ otherwise. Thus, $\overline{b}_{1} = b_{1}$, $\overline{b}_{2} = b_{2}$, $\overline{b}_{6} = b_{8}$, and 
$\overline{b}_{8} = b_{6}$. Note that  $b_{3} = t$ \emph{is just a notation}, it can (and actually will) evolve. 
From $\sigma_{1}(\mathcal{H}_{f} - \mathcal{E}_{5}) = \mathcal{E}_{7}$ we see that 
$(\overline{f},\overline{g})(0,g) = \left(\frac{B}{D},- \frac{K}{M}\right) = (1,\infty)$ and so $B = D$ and $M=0$.
From $\sigma_{1}(\mathcal{H}_{f} - \mathcal{E}_{7}) = \mathcal{E}_{5}$ we see that 
$(\overline{f},\overline{g})(1,g) = \left(\frac{A+B}{C+D},\frac{L}{M}\right) = (0,0)$ and so $A = -B$ and $L=0$. Finally, from 
$\sigma_{1}(\mathcal{E}_{1}) = \mathcal{E}_{1}$ we see that 
$(\overline{f},\overline{g})(\infty,b_{1}) = \left(\frac{A}{C},\frac{K b_{1}}{N}\right) = (\infty,b_{1})$; so $C = 0$ and $K=N$. Thus, 
\begin{equation*}
	\overline{f} = 1 - f,\qquad \overline{g} = \frac{(f-1)g}{f}.
\end{equation*}
Finally, from $\sigma_{1}(\mathcal{E}_{3}) = \mathcal{E}_{3}$ we see that 
$(\overline{f},\overline{g})(b_{3},\infty) = (\overline{f},\overline{g})(t,\infty) = (1-t,\infty) = (\overline{t},\infty)$, and so we see that 
the parameter $t$ indeed evolves, $\overline{t} = 1-t$. This is related to the fact that elements from $\operatorname{Aut}\left(D_{4}^{(1)}\right)$
are no longer standard B\"acklund transformations of $\Pain{VI}$. This completes the proof of \eqref{eq:sigma1}.
The proof for the other $\sigma$s is similar and is omitted. 
\end{proof}

Finally, the semi-direct product structure is defined by the action of $\sigma \in \operatorname{Aut}\left(D_{4}^{(1)}\right)$ on 
$W\left(D_{4}^{(1)}\right)$ via $w_{\sigma(\alpha_{i})} = \sigma w_{\alpha_{i}} \sigma^{-1}$.



\subsection{The standard discrete d-$\Pain{V}$ Painlev\'e Equation} 
\label{sub:the_standard_discrete_d_p__text_v_painlev_e_equation}
As is well-know, there are infinitely many different discrete Painlev\'e equations of the same type, since they correspond to the non-conjugate 
translations in the affine symmetry sub-lattice $Q$. Some of these equations are special, since they either appear in applications, or have
a particularly nice form, or have degenerations to other known equations. In the d-$\dPain{D_{4}^{(1)}/D_{4}^{(1)}}$ family one such 
equation is known as a difference Painlev\'e-V equation, since it has a continuous limit to the differential Painlev\'e-V equation.

In \cite{KajNouYam:2017:GAOPE} this equation is given in the following form,
\begin{equation}\label{eq:dPv-std}
	\overline{f} f = \frac{t g (g - a_{4})}{(g + a_{2}) (g + a_{1} + a_{2})},\qquad 
	g + \underline {g} = a_{0} + a_{3} + a_{4} + \frac{a_{3}}{f - 1} + \frac{t a_{0}}{ f - t},
\end{equation}
with the root variable evolution and normalization given by 
\begin{equation}\label{eq:dPv-rv-evol}
	\overline{a}_{0} = a_{0} - 1, \quad \overline{a}_{1} = a_{1}, \quad \overline{a}_{2} = a_{2} + 1,\quad \overline{a}_{3} = a_{3} - 1,\quad
	\overline{a}_{4} = a_{4},\qquad a_{0} + a_{1} + 2 a_{2}  + a_{3} + a_{4} = 1.
\end{equation}
From the root variable evolution \eqref{eq:dPv-rv-evol} we immediately see that the corresponding translation on the root lattice is 
\begin{equation}\label{eq:dPv-transl}
	\varphi_{*}: \upalpha =  \langle \alpha_{0}, \alpha_{1}, \alpha_{2}, \alpha_{3}, \alpha_{4}  \rangle
	\mapsto \varphi_{*}(\upalpha) = \upalpha + \langle 1,0,-1,1,0 \rangle \delta.
\end{equation}
Using the standard techniques, see \cite{DzhTak:2018:OSAOSGTODPE} for a detailed example, we get the following decomposition of $\psi$ in 
terms of the generators of $\widetilde{W}\left(D_{4}^{(1)}\right)$:
\begin{equation}\label{eq:dPv-decomp}
	\varphi = \sigma_{3}\sigma_{2} w_{3} w_{0} w_{2} w_{4} w_{1} w_{2}. 
\end{equation}

\section*{Acknowledgements} 
\label{sec:acknowledgements} AD acknowledges the support of the University of Northern Colorado Spring~2019 RSCW grant and the support 
of the MIMUW grant   
to visit Warsaw in August 2019; that visit was essential for the success of the project. GF acknowledges the support of the 
National Science Center (Poland) via grant OPUS 2017/25/B/BST1/00931. 
AS was supported by a University College London Graduate Research Scholarship and Overseas Research Scholarship.

Part of this work was done at the 15th International Symposium on Orthogonal Polynomials, Special Functions and Applications (OPSFA-19) 
in Hagenberg, Austria, and AD and GF would like to thank the organizers for an interesting and stimulating Conference. 

We also thank Kenji Kajiwara, Tomoyuki Takenawa, and Yasuhiko Yamada for helpful comments and discussions.

\bibliographystyle{amsalpha}

\providecommand{\bysame}{\leavevmode\hbox to3em{\hrulefill}\thinspace}
\providecommand{\MR}{\relax\ifhmode\unskip\space\fi MR }
\providecommand{\MRhref}[2]{%
  \href{http://www.ams.org/mathscinet-getitem?mr=#1}{#2}
}
\providecommand{\href}[2]{#2}

\end{document}